\documentclass[12pt,english,american]{article}

\usepackage[colorlinks=true,linkcolor=Mahogany,citecolor=NavyBlue]{hyperref}

\usepackage{amsmath,amsfonts,amssymb,amsthm,mathrsfs}

\usepackage{setspace}

\usepackage{enumerate}
\usepackage{bm}
\usepackage{bbm}
\usepackage{verbatim}
\usepackage[usenames,dvipsnames]{color}
\usepackage{xargs} 
\usepackage{authblk}
\usepackage[square,sort,comma,numbers]{natbib}

\newcommand{\field}[1]{\mathbb{#1}}
\newcommand{\N}{\field{N}}

\newcommand{\Z}{\field{Z}}

\newcommand{\HH}{\mathscr H}

\newcommand{\LL}{\mathscr L}

\newcommand{\eps}{\varepsilon}

\newcommand{\sprod}[2]{\mbox{$\langle #1,#2 \rangle$}}   % scalar product
\newcommand{\ket}[1]{\left| #1 \right\rangle}   % ket vector
\newcommand{\Nmu}{N\mspace{-1mu}(\mspace{-1mu}\mu\mspace{-1mu})}

\newcommand{\rst}{\! \upharpoonright \!} %eingeschraenkt auf

\theoremstyle{plain}\newtheorem{theorem}{Theorem}[section]
\theoremstyle{plain}\newtheorem{lemma}[theorem]{Lemma}
\theoremstyle{plain}\newtheorem{corollary}[theorem]{Corollary}
\theoremstyle{plain}
\theoremstyle{plain}\newtheorem{proposition}[theorem]{Proposition}
\theoremstyle{definition}
\theoremstyle{definition}
\theoremstyle{definition}\newtheorem*{remarks}{Remarks}
\theoremstyle{definition}\newtheorem{def:and:lemma}[theorem]{Definition and Lemma}

\newcommand{\lsp}{\big \langle }
\newcommand{\rsp}{\big \rangle }

\newcommand{\D}{\textnormal{d}}

\newcommand{\scp}[2]{\left\langle #1 , #2 \right\rangle}

\newcommand{\norm}[2][]{\big | \hspace{-0.2mm} \big | #2 \big | \hspace{-0.2mm} \big |_{#1}}
\newcommand{\snorm}[2][]{| \hspace{-0.2mm}| #2 | \hspace{-0.2mm} |_{#1}}

\newcommand{\T}[1]{}

\numberwithin{equation}{section}

\title{{Weak coupling limit for the ground state energy\\ of the 2D Fermi polaron}}

\author{David Mitrouskas}  

\begin{document}

\maketitle

\frenchspacing

\begin{spacing}{1.15} 

\begin{abstract}
We analyze the ground state energy for $N$ fermions in a two-dimensional box interacting with an impurity particle via two-body point interactions. We allow for mass ratios $M >1.225$ between the impurity mass and the mass of a fermion and consider arbitrarily large box sizes while keeping the Fermi energy fixed. Our main result shows that the ground state energy in the limit of weak coupling is given by the polaron energy. The polaron energy is an energy estimate based on trial states up to first order in particle-hole expansion, which was proposed by Chevy in the physics literature. For the proof we apply a Birman--Schwinger principle that was recently obtained by Griesemer and Linden. One main new ingredient is a suitable localization of the polaron energy.
\end{abstract}
%\tableofcontents

%---------------------------------------------------------------------------------------------------------------------------------------------------------------
\section{Introduction and main result}
%---------------------------------------------------------------------------------------------------------------------------------------------------------------

It is a universal challenge in quamtum theory to understand the physics of few particles immersed into a complex enviroment in terms of properties of quasi-particles. A famous example of a quasi-particle is the Fr\"ohlich polaron developed in a series of influential works by Landau, Pekar and Fr\"ohlich \cite{Landau33,LandauPekar,Pekar1954,Froehlich1954}. They suggested to describe the motion of an electron through a polarizable crystal in terms of a polaron, that is, a quasi-particle composed of an electron dressed by a local deformation of the crystal. This picture leads to a drastic simplification as the complex many-body problem is replaced by a self-consistent non-linear one-body model, which is much more accessible to computations. While the Fr\"ohlich polaron is certainly the most prominent example for a polaron model, the concept of quasi-particles and polarons has turned out very useful far beyond its original application in the theory of electrons moving through crystals. For instance the experimental realization of impurities immersed into ultracold atomic gases during the last two decades has triggered the invention and analysis of many new models such as the Fermi polaron \cite{Chevy_06}, the Bose polaron \cite{Grusdt_Demler15} and the angulon \cite{SchmidtL15}. In the present work we are interested in the two-dimensional Fermi polaron which is a popular model in theoretical physics to describe strongly population imbalanced Fermi gases at low temperature confined to the two-dimensional plane. In case of extreme imbalance there is only a single particle interacting with a gas of non-interacting fermions via a two-body short range interaction. 

We consider $N$ identical fermions and an additional distinguished particle, called impurity particle, in a two-dimensional box $\Omega =[-L/2,L/2]^2$ with periodic boundary conditions. The underlying Hilbert space is $L^2(\Omega) \otimes \HH_{N}$ where $\HH_{N} = \bigwedge^N L^2(\Omega)$ denotes the space of anti-symmetric $N$-particle wave functions. For a short-range potential, the Pauli principle suppresses the interaction among the fermions which is therefore neglected. The Hamiltonian of the system is formally described by 
\begin{align}\label{eq: formal definition of H}
 -\frac{1}{M}\Delta_y -\sum_{i=1}^N \Delta_{x_i} - g \sum_{i=1}^N \delta(x_i-y),
\end{align}
where $y$ represents the coordinate of the impurity, $\Delta$ is the Laplace operator and $M$ denotes the ratio between the mass of the impurity particle and the mass of a fermion. The interaction is given by a Dirac-delta-potantial $\delta(x)$ with coupling strength $g>0$. This model is known as the 2D Fermi polaron and has been analyzed to a great extent in the physics literature, see e.g. \cite{Chevy_06,Combescot_Giraud_08,ProkofevS08,ChevyMora_09,PDZ_09,Bruun_Massignan10,Parish11,Schmidt_Enss_2012,Parish_Levinsen13}. The Fermi polaron is of interest, among other reasons, because of the occurrence of a pairing mechanism somewhat analogous to the famous BCS--BEC crossover. In two space dimensions, one expects a transition of the ground state as a function of the coupling strength. While for weak coupling, the impurity particle is expected to be surrounded by a cloud of particle-hole excitations, in the strong coupling regime it is predicted that the impurity is closely bound by a single fermion forming a molecular state.

Here we provide a rigorous analysis of the ground state energy in the limit of weak coupling by which we confirm its asymptotic form conjectured in the physics literature. From the mathematical point of view, our work is a continuation of recent articles by Griesemer and Linden \cite{Linden,GL_Stability,GL_Variational} in which they provide a definition of the self-adjoint Hamiltonian $H$ associated with the formal expression \eqref{eq: formal definition of H}, derive a Birman-Schwinger type principle for this Hamiltonian and prove stability of the Fermi polaron at zero density. The Birman--Schwinger priciple characterizes the low energy spectrum by means of an operator $\phi(\lambda)$ with spectral parameter $\lambda$. Compared to $H$ the operator $\phi(\lambda)$ is given more explicitly and thus provides a suitable tool for the analysis of the low energies, in particular for upper and lower bounds for the ground state energy $\inf \sigma (H)$. Two such upper bounds, called polaron and molecule energy, respectively, were discussed in \cite{GL_Variational}. Motivated by the derivation of these upper bounds, we shall provide a matching lower bound for $\inf \sigma (H)$ in the limit of weak coupling.

%---------------------------------------------------------------------------------------------------------------------------------------------------------------
\subsection{The model}
%---------------------------------------------------------------------------------------------------------------------------------------------------------------

A possible approach to define the Fermi polaron is to start with a regularized version of the point interaction and then remove the regularization in a suitable sense \cite{GL_Variational}. Since this lays the foundation for our work, we provide a short summary.

For reasons of convenience we describe the fermions in the formalism of second quantization. This means that we think of $\HH_{N}$ as the $N$-particle sector of $\mathcal F = \bigoplus_{n=0}^\infty \bigwedge^n L^2(\Omega)$, the fermionic Fock space over $L^2(\Omega)$. We denote the vacuum state with zero particles by $\ket{0} = (1,0,0,...)$ and define creation and annihilation operators $a_k^*,a_k : \mathcal F \to \mathcal F$ of plane waves $\varphi_k(x) = L^{-1} e^{ikx} $, $k\in  (2 \pi / L) \mathbb Z^2$,
\begin{align}
(a_k \Psi)^{(n)} & = \sqrt{n+1} \int_{\Omega} \D x_{n+1} \overline{\varphi_k(x_{n+1})} \Psi^{(n+1)}(x_1,...,x_{n+1}),  \\
(a_k^* \Psi)^{(n)} & = \frac{1}{\sqrt n }\sum_{j=1}^n (-1)^{j}\varphi_k(x_j) \Psi^{(n-1)}(x_1,...,x_{j-1},x_{j+1},...,x_n)
\end{align}
for $\Psi = (\Psi^{(n)})_{n\ge 0} \in \mathcal F$. The creation and annihilation operators satisfy the usual canonical anti-commutation relations (CAR),
\begin{align}
 a_k a^*_l + a_l^* a_k = \delta_{kl},\quad  a_k a_l + a_l a_k  = 0 
\end{align}
for all pairs $k,l\in (2\pi/L) \mathbb Z^2$.

For any number $E_B<0$ we introduce the inverse coupling constant 
\begin{align}\label{eq: renormalizaton condition}
g_n^{-1} =  \sum_{k^2\le n} \frac{1}{ (1+\frac{1}{M})  k^2-E_B}
\end{align}
and define the sequence of regularized Hamiltonians $(H_n)_{n\in \mathbb N}$, acting on $L^2(\Omega)\otimes \HH_{N}$, by
\begin{align}\label{eq: regularized Hamiltonian}
H_n  =  -\frac{1}{M} \Delta_y + \sum_{k} k^2   a_k^* a_k  - g_n  \sum_{k^2,l^2\le n }\, e^{i(k-l)y}\, a_l^* a_k.
\end{align}
If not stated otherwise, sums run over the two-dimensional momentum lattice $(2\pi /L) \mathbb Z^2$ with possible restrictions indicated, e.g., as $k^2\le n$. 

The following statement proves the existence of the self-adjoint Hamiltonian describing the 2D Fermi polaron.

\begin{proposition}\textnormal{(see \cite[Theorem 6]{GL_Variational})} \label{prop: def of the Hamiltonian} For given $L>0$, $N\ge 1$, $M>0$ and $E_B<0$ there exists a self-adjoint Hamiltonian $H:D(H) \subseteq L^2(\Omega) \otimes \HH_N \to   L^2(\Omega) \otimes \HH_N  $ such that $H_n \to H $ in strong resolvent sense as $n\to \infty$. $H$ is bounded from below. 
\end{proposition}
\begin{remarks}{\color{white}{a}}\medskip

\noindent \textbf{1.1.} From Proposition 5.1 \cite{GL_Variational} we know that the spectrum $\sigma (H)$ is purely discrete.\medskip

\begin{comment}
\noindent \textbf{2.1.} To see that Theorem 6 \cite{GL_Variational} applies to the regularized Fr\"ohlich polaron written as in \eqref{eq: regularized Hamiltonian}, let us note that the latter reads in second quantization as
\begin{align}\label{eq: Regularized Hamiltonian second quantization}
H_n =     \sum_{k} k^2  ( a_k^* a_k +\frac{1}{M}b^*_k b_k)  - g_n \frac{1}{L^2} \sum_{k,l,q} \beta_n(k) \beta_n(l)  \, b^*_{l-q} b_{k-q}\, a_l^* a_k,
\end{align}
with $\beta_n$ the characteristic function $\beta_n(k) =1 $ for all $k^2\le n$ and zero otherwise. Now this expression is a special case of the Hamiltonian $H_n$ that is considered in Theorem 6 \cite{GL_Variational}. That \eqref{eq: Regularized Hamiltonian second quantization} satisfies all required conditions of the latter was explained in Remark 1 directly thereafter.\medskip
\end{comment}
\noindent \textbf{1.2.} The choice of $g_n$ ensures the following renormalization condition (note that $g_n$ has a logarithmic divergence as $n\to \infty$): For $N=1$ the Hamiltonian $H$ has exactly one negative eigenvalue which coincides with $E_B<0$. Hence the number $E_B$ corresponds to the binding energy of the $1+1$-particle model and can be used as a suitable coupling parameter of the point interaction.

\end{remarks}

The goal of this work is to derive an asymptotic formula for the ground state energy $\min \sigma (H)$ in the limit of weak coupling $E_B \nearrow 0$ or in the limit of large density $NL^{-2}\to \infty$ (the two limits turn out to be closely connected). Instead of working with the particle number $N$ as a free parameter, it is more convenient to fix a chemical potential $\mu >0$ and then choose the number of fermions by $N=N(\mu)$ with
\begin{align}\label{DEF:N(MU)}
N(\mu) = \left\vert \left\{ k \in  \kappa \mathbb  Z^2\, : \, k^2 \le \mu \right\} \right\vert, \quad \kappa = \frac{2\pi}{L}.
\end{align}
Since the number of fermions now coincides with the number of eigenvalues of $-\Delta$ that are less or equal than $\mu$, counting multiplicities, the parameter $\mu$ plays the role of the Fermi energy. We write the ground state energy of $H$ as a function of $\mu$ and $E_B$ as $E (\mu,E_B) =\min \sigma (H )$ and introduce the energy of $N(\mu)$ non-interacting fermions inside the box $\Omega$,\footnote{The energies $E (\mu,E_B)$ and $E_0(\mu)$ depend of course also on $L$ but we omit this in our notation.}
\begin{align}
E_0(\mu) =\sum_{k^2\le \mu} k^2.
\end{align}

Our main result, Theorem \ref{thm: main theorem}, shows that the energy difference $E (\mu,E_B)-E_0(\mu)$ is given at leading order by the polaron energy $e_{\rm P}(\mu,E_B)$, that is
\begin{align}\label{eq: main result introduction}
\frac{ E (\mu,E_B) - E_0(\mu)}{ e_{\rm P}(\mu,E_B) } =  1 +  o (1) \quad \text{as}\quad \frac{\mu}{\vert E_B\vert }\to \infty.
\end{align} 
The polaron energy $e_{\rm P}(\mu,E_B)<0$ is the lowest solution to the polaron equation
\begin{align}\label{eq: polaron equation}
e_{\rm P}(\mu,E_B) = - \frac{1}{L^2}\sum_{k^2\le \mu} \frac{1}{G(k,-k^2 - e_{\rm P}(\mu,E_B))},
\end{align}
where $G(q,\tau)$ is defined for $\tau>-\mu$ and $q \in \kappa \mathbb Z^2$ by
\begin{align}\label{eq: definition of G}
G(q,\tau) = \frac{1}{L^2}\sum_{k} \left( \frac{1}{ (1+\frac{1}{M})  k^2-E_B} - \frac{\chi_{(\mu,\infty)}(k^2)}{\frac{1}{M}(q-k)^2 + k^2 +\tau}\right).
\end{align}
Here $\chi_{(\mu,\infty)}(s)$ denotes the characteristic function $\chi_{(\mu,\infty)}(s)=1$ for $s > \mu$ and $\chi_{(\mu,\infty)}( s )=0$ otherwise. That \eqref{eq: polaron equation} admits a lowest negative solution was shown in \cite[Proposition 7.1]{GL_Variational}. Let us mention that our main result \eqref{eq: main result introduction} holds in particular in the thermodynamic limit, i.e. after taking the limit $L\to \infty$.

The polaron equation \eqref{eq: polaron equation} was proposed in \cite{Chevy_06} based on a formal variational calculation with trial states $w_{\rm P} \in L^2(\Omega) \otimes \HH_{N(\mu)}$ of the form
\begin{align}\label{eq: polaron state}
w_{\rm P} = \alpha_0 \varphi_0 \otimes \ket{\rm {FS}_\mu } + \sum_{k ^2\le \mu }\sum_{l^2> \mu} \alpha_{k , l} \varphi_{k-l} \otimes a^*_l a_k   \ket{\rm {FS}_\mu } 
\end{align}
where $\alpha_0,\alpha_{k,l}\in \mathbb C$, $\varphi_k(y) = L^{-1}e^{iky}$ and
\begin{align}
\ket{\rm {FS}_\mu }  = \prod_{k^2\le \mu} a_k^* \ket{0}
\end{align}
denotes the ground state of the kinetic operator $\sum_k k^2 a_k^* a_k \restriction \HH_{N(\mu)}$ (called the Fermi sea). A rigorous proof of the upper bound $E(\mu,E_B) \le E_0(\mu) + e_{\rm P}(\mu,E_B)$ was given in \cite{GL_Variational} utilizing a generalized Birman--Schwinger principle for the Hamiltonian $H$ (see Section \ref{sec: preliminaries}).

In the physics literature the polaron energy is considered to be a good approximation in the weak coupling limit $E_B \nearrow 0$ as well as in the large density limit $\mu\to \infty$ \cite{Chevy_06,Combescot_Giraud_08,Parish11}. In the regime of strong coupling $E_B\to -\infty$, it is expected that the ground state undergoes a transition to states in which the impurity is tightly bound by a single fermion. This behavior is represented by the so-called molecule or dimer ansatz \cite{ChevyMora_09,PDZ_09,Parish11}. In contrast to the latter, the polaron state \eqref{eq: polaron state} is interpreted as an impurity that is surrounded by weak density fluctuations in the Fermi sea. The two classes of trial states were investigated extensively in the physics literature leading to indications for the anticipated difference between the shape of the ground state in the weak and strong coupling limits (see, e.g., the literature quoted in the previous section). For this reason the Fermi polaron is also discussed in the context of the BCS--BEC crossover. Most results in the physics literature, however, are based on variational estimates using suitable classes of trial states. We remark that this can only justify upper bounds for the ground state energy, whereas here we provide a corresponding lower bound.

The Fermi polaron has been studied also in three dimensions. The problem of defining a semi bounded self-adjoint Hamiltonian in this case was solved in \cite{Minlos11,CDFMT12,Moser_Seiringer}. Contrary to the 2D model, it is known that the Hamiltonian is semi-bounded in three dimensions only if $M\ge M_*$ for some critial mass ratio $M_*>0$. Rigorous results concerning the ground state energy mostly addressed the question of stability and the existence of a lower bound that is uniform in the particle number $N$. In \cite{Moser_Seiringer} it was shown that at zero density there is such a uniform lower bound under the condition that $M>0.36$. In a more recent work, Moser and Seiringer generalized their findings to the positive density setup by proving that the energy shift caused by the impurity particle depends only on the average density and the interaction strength but not on the size of the system \cite{Moser_Seiringer_2}. The question whether the polaron energy describes the correct asymptotic form of the ground state energy similar to \eqref{eq: main result introduction} is still open for the three-dimensional model.

Quantum models with $N+1$ particles interacting via two-body point interactions have been studied in the mathematical literature from various points of views. Besides the works already quoted, we refer to \cite{DellAntonioEtAl,Dimock_Rajeev_04,Solvable_Models,Minlos11,CDFMT12,MichelangeliO_18} and references therein.

\subsection{Main result}

We are now ready to state our main result which provides an asymptotic estimate for the ground state energy $\min \sigma(H)$ of the 2D Fermi polaron.

\begin{theorem}\label{thm: main theorem}
Set $M>1.225$ and for $L > 0$ and $\mu > 0$, fix the number of particles $N(\mu)$ by \eqref{DEF:N(MU)}. Moreover, let the Hamiltonian $H$ be the limit operator of $(H_n)_{n\in \mathbb N}$ as stated in Proposition \ref{prop: def of the Hamiltonian}. Then the ground state energy $E (\mu,E_B) = \min \sigma (H)$ and the lowest solution $e_{\rm P}(\mu,E_B) < 0$ of the polaron equation \eqref{eq: polaron equation} satisfy the following property. There exist constants $c_0, C > 0$ (possibly depending on $M$) such that
\begin{align}\label{eq: main estimate}
\big\vert  E (\mu,E_B)  -  E_0(\mu) - e_{\rm P}(\mu,E_B) \big\vert \, & \le C \,  \frac{|e_{\rm P}(\mu,E_B)|}{\log(\mu / |E_B|)} 
\end{align}
for all $L>0$, $\mu>0$ and $E_B<0$ with $L^2 \vert E_B\vert \ge 1$ and $\mu / \vert E_B\vert \ge c_0$.
\end{theorem}

\begin{remarks} {\color{white}{a}}\medskip

\noindent \textbf{1.3.} In Lemma \ref{lemma: asymptotic polaron energy} we show that $ e_{\rm P}(\mu,E_B) = O(\mu / \log(\mu/ \vert E_B \vert ))$ as $\mu/\vert E_B\vert  \to \infty$.\medskip

\noindent \textbf{1.4.} The condition $L^2\vert E_B \vert\ge 1$ characterizes the range of parameters in which the two-body binding energy $E_B$ is at least of the order of the minimal kinetic excitation energy which equals $(2\pi/L)^2$. In this sense our analysis is beyond the perturbative regime.\medskip

\noindent \textbf{1.5.} Since the constant on the right side of \eqref{eq: main estimate} does not depend on $L>0$, we can directly infer a statement about the ground state energy in the thermodynamic limit,
\begin{align}\label{eq: main estimate TD limit}
\limsup_{L \to \infty } \bigg\vert \frac{ E (\mu,E_B)  -  E_0(\mu)}{ e_{\rm P}(\mu,E_B)} - 1\bigg \vert \, & \le  \frac{C}{\log(\mu / |E_B|)} 
\end{align}
for all $\mu / \vert E_B\vert \ge c_0$.\medskip

\noindent \textbf{1.6.} The condition $M>1.225$ is related to the problem of stability (of second kind), that is, to find a uniform lower bound for the ground state energy in the thermodynamic limit $L\to \infty$. While it is known that $H$ is bounded from below for all $M>0$ \cite{DellAntonioEtAl,GL_Variational}, it is unclear whether a uniform bound exists when $M \le 1.225$. This is an unsolved problem also in the case of zero density, see \cite{GL_Stability}.\medskip

\noindent \textbf{1.7.} The upper bound in \eqref{eq: main estimate} was proven in \cite{GL_Variational}. For the convenience of the reader, we give a brief sketch of the argument in Section \ref{sec: upper bound}. The novel contribution of the present work is the derivation of the lower bound.\medskip

\noindent \textbf{1.8.} A similar result was obtained in \cite{LM_InfiniteMass} for the case of an infinitely heavy impurity, formally corresponding to $M= \infty$. In this case the $N$ fermions interact with an external delta potential which simplifies the analysis significantly.
\end{remarks}

The rest of the article is organized as follows. In the next section we introduce the Birman--Schwinger operator $\phi(\lambda)$ associated to the Hamiltonian $H$ and state the corresponding Birman--Schwinger principle. Upper and lower bounds for the ground state energy follow from suitable bounds for $\phi(\lambda)$. In Section \ref{sec: upper bound} we recall how to obtain the upper bound in \eqref{eq: main estimate}. Sections \ref{sec: localization lemma}--\ref{sec: asymptotic form of G} are about the matching lower bound. They account for the main part of this work. In Section \ref{sec: localization lemma} we derive a localization of the polaron energy inside a suitable subspace of the Hilbert space. In the two subsequent sections we provide lower bounds for the Birman--Schwinger operator on the localization subspace and its orthogonal complement. On the localization subspace, we obtain a perturbed polaron equation whose solution we compare to the polaron energy, see Section \ref{sec: analysis of Phi}. In Section \ref{sec: analysis of Psi} we analyze the Birman--Schwinger operator on the orthogonal complement of the localization subspace. The lower bound on this subspace can be understood as a proof of stability of the Fermi polaron at positive density which generalizes analogous findings for the zero density model \cite{GL_Variational}. In Section \ref{sec: proof of main theorem} we combine the obtained results to conclude the proof of Theorem \ref{thm: main theorem}. The last section contains the proof of a technical lemma that is used several times throughout the article.

%---------------------------------------------------------------------------------------------------------------------------------------------------------------
\section{Preliminaries and upper bound\label{sec: preliminaries}}
%---------------------------------------------------------------------------------------------------------------------------------------------------------------

In this section we discuss the Birman--Schwinger principle for the Hamiltonian $H$ which provides a suitable tool for the analysis of upper and lower bounds for $E(\mu,E_B) = \min \sigma (H)$.

 %---------------------------------------------------------------------------------------------------------------------------------------------------------------
\subsection{The Birman--Schwinger operator $\phi(\lambda)$ \label{sec: BS principle and upper bound}}
%---------------------------------------------------------------------------------------------------------------------------------------------------------------

Our starting point for the proof of Theorem \ref{thm: main theorem} is a Birman--Schwinger type principle for the operator $H$. This is the second result from \cite{GL_Variational} which is important for our analysis. For the precise statement, let us introduce the resolvent set $\rho(H_0) \subset \mathbb C$ of the non-interacting Hamiltonian
\begin{align}
H_0 = \Big( - \frac{1}{M} \Delta_y + T \Big) \restriction L^2(\Omega)\otimes \HH_{N(\mu)},
\end{align}
with $T = \sum_k k^2 a_k^* a_k $ the kinetic energy operator on the fermionic Fock space.

\begin{proposition}\textnormal{(See \cite[Sections 5 and 6]{GL_Variational})} \label{prop: Birman-Schwinger} There exists a family of operators $\phi(\lambda)$, $\lambda \in \rho(H_0)$, acting on $L^2(\Omega) \otimes \HH_{N(\mu)-1}$ with $\lambda$-independent domain $\mathscr D$, such that for all real-valued $\lambda$, $\phi(\lambda)$ is essentially self-adjoint and its closure (denoted again by $\phi(\lambda)$) satisfies 
\begin{align}\label{eq: Birmann Schwinger principle}
\inf \sigma( \phi (\lambda) )  \le  0  \qquad \Leftrightarrow \qquad E(\mu,E_B)  \le  \lambda,
\end{align}
with equality on one side implying equality on both sides. Moreover the $\phi(\lambda)$ form an analytic family of type (A) and for $\lambda \in \mathbb R_-\cup (\mathbb C \setminus \mathbb R)\subset \rho(H_0)$ they are given explicitly by
\begin{align}\label{def: Birman-Schwinger operator}
\phi(\lambda) & =  F( i\nabla_y ,T-\lambda)  + \frac{1}{L^2} \sum_{k,l} a_l^* e^{i ky} \frac{1}{- \frac{1}{M} \Delta_y + T     + k^2 + l^2  -\lambda} e^{ - i l y} a_k
\end{align}
where
\begin{align}
 F( q ,\tau ) = \frac{1}{L^2}\sum_{k}\bigg( \frac{1}{mk^2 -  E_B } - \frac{1}{\frac{1}{M} q^2 + k^2 + \tau}\bigg), \quad m = \frac{M+1}{M}.
\end{align}

\end{proposition}

\noindent\textbf{Remarks.}\medskip

\noindent \textbf{2.1} Note that while $H$ is defined on the Hilbert space $L^2(\Omega)\otimes \HH_{N(\mu)}$, the Birman--Schwinger operator $\phi(\lambda)$ acts on $L^2(\Omega)\otimes \HH_{N(\mu)-1}$. We also remark that the domain $\mathscr D$ is given by the set of all finite linear combinations of states of the form $\varphi_q \otimes \varphi_{k_1} \wedge \ldots \wedge \varphi_{k_{N(\mu)-1}}$ with $q,k_1,...,k_{N(\mu)-1} \in \kappa \mathbb Z^2$ and $\varphi_k$ the normalized plane waves in $L^2(\Omega)$.\medskip

\noindent \textbf{2.2.} The operator defined in \eqref{def: Birman-Schwinger operator} coincides with the Birman--Schwinger operator $\phi(z)$ from \cite[Lemma 6.3]{GL_Variational} up to a multiplicative factor $L^{-2}$. Apart from renaming $z$ into $\lambda$, we write the impurity degree of freedom in first quantization whereas in \cite{GL_Variational}, all degrees of freedom are expressed in second quantization. Proposition \ref{prop: Birman-Schwinger} is a direct consequence of the statements from \cite[Section 5 and Lemma 6.3]{GL_Variational}.\medskip

For explicit computations it is useful to invert the normal order of creation and annihilation operators in \eqref{def: Birman-Schwinger operator} when $k^2,l^2\le \mu$. With $G(q,\tau)$ defined in \eqref{eq: definition of G}, this leads for $\lambda \in \mathbb R_-\cup (\mathbb C \setminus \mathbb R)$ to\allowdisplaybreaks
\begin{align}\label{eq: phi(lambda) for all lambda}
 \phi(\lambda) & =  G(i\nabla_y,T-\lambda)  - \frac{1}{L^2} \sum_{k^2,l^2 \le \mu} a_k e^{i ky} \frac{1}{- \frac{1}{M} \Delta_y + T    -\lambda} e^{ - i l y} a_l^*\notag \\
& \quad\quad\quad - \bigg(\frac{1}{L^2}  \sum_{\substack{k^2\le \mu\\ l^2>\mu} } e^{i ky}a_k a_l^*  \frac{1}{- \frac{1}{M} \Delta_y + T +l^2    -\lambda} e^{ - i l y} + \text{h.c.}\bigg) \notag \\
& \quad\quad\quad + \frac{1}{L^2} \sum_{k^2,l^2 > \mu} a_l^*  e^{i ky} \frac{1}{- \frac{1}{M} \Delta_y + T   + k^2 + l^2 -\lambda} e^{ - i l y} a_k
\end{align}
understood as an operator on $L^2(\Omega) \otimes \HH_{N(\mu)-1}$. Through analytic continuation the above identity extends to $\lambda < E_0(\mu)$. This explicit expression of $\phi(\lambda)$ will be the main object to be analyzed.

To arrive at \eqref{eq: phi(lambda) for all lambda} we made use of the CAR and the pull-through formula, which for suitable functions $f: \kappa \mathbb Z^2 \times \mathbb R \to \mathbb C$ reads
\begin{align}\label{eq: pull through formula}
a_k f(P_{\rm f},T)    & = f(P_{\rm f}+k,T+k^2) a_k    , \quad a_k^* f(P_{\rm f},T)   = f(P_{\rm f} -k,T-k^2) a_k^*   .
\end{align}
Here $P_{\rm f} = \sum_{k} k a_k^* a_k$ denotes the momentum operator of the fermions.

\subsection{Upper bound\label{sec: upper bound}}

We show how to use Proposition \ref{prop: Birman-Schwinger} to obtain an upper bound for $E(\mu,E_B)$. This resembles the analysis performed in the first part of Section 7 \cite{GL_Variational}. For an upper bound, it is sufficient to find a trial state $w$ and a suitable $\lambda$ that satisfy $\langle w, \phi(\lambda) w \rangle \le 0$. As such we choose $\lambda = E_0(\mu)+e_{\rm P}(\mu,E_B)$ and the wave function
\begin{align}\label{eq: trial state}
w =   \sum_{k^2\le \mu}  \frac{1}{G(k, -k^2 - e_{\rm P}(\mu,E_B) ) }  \varphi_k \otimes a_k \ket{ {\rm FS}_\mu } .
\end{align}
With the aid of \eqref{eq: phi(lambda) for all lambda}, a straightforward computation leads to
\begin{align}
& \lsp w, \phi ( E_0(\mu)+e_{\rm P}(\mu,E_B) ) w \rsp \notag \\
& =  \sum_{k^2\le \mu} \frac{1}{ G(k, -k^2 -e_{\rm P}(\mu,E_B)     )  } \bigg[ 1 +   \frac{1}{L^2} \sum_{k^2 \le \mu} \frac{1}{G(k, -k^2 -e_{\rm P} (\mu,E_B) )} \cdot   \frac{1} { e_{\rm P} (\mu,E_B) } \bigg] ,
\end{align}
which is identically zero because of \eqref{eq: polaron equation}. By Proposition \ref{prop: Birman-Schwinger} this implies the upper bound 
\begin{align}
E(\mu,E_B)  \le E_0(\mu) + e_{\rm P}(\mu,E_B).
\end{align}

\subsection{Momentum decomposition of $\phi(\lambda)$}

For the analysis of the lower bound it is convenient to make use of the translational invariance of the model, in particular, that $\phi(\lambda)$ commutes with the total momentum operator $P_{\rm tot} = -i\nabla_y + P_{\rm f}$ with $P_{\rm f} = \sum_{k} k a_k^* a_k \restriction \HH_{N(\mu)-1}$. This guarantees a total momentum decomposition of $\phi(\lambda)$, meaning that there is a unitary map 
\begin{align}
V: L^2(\Omega)\otimes \HH_{N(\mu)-1} \to \bigoplus_{p\in \kappa \mathbb Z^2} \HH_{N(\mu)-1}
\end{align}
that diagonalizes $P_{\rm tot}$ by eliminating the $y$ coordinate in favor of the total momentum $p\in \kappa  \mathbb Z^2$. This unitary is given by
%$V w = \sum_{p\in \kappa \mathbb Z^2} (Vw)_p$, 
$ (Vw)_p =  ( \langle \varphi_p\vert \otimes 1_{\HH_{N(\mu)-1}} ) \, e^{iP_{\rm f}y} w $ where $ \langle \varphi_p\vert \otimes 1_{\HH_{N(\mu)-1}}$ shall indicate to take the scalar product in the coordinate $y$ with the plane wave $\varphi_p\in L^2(\Omega) $. To see that the parameter $p$ describes the total momentum, use $(V P_{\rm tot} w)_p = p (Vw)_p$ to verify
\begin{align}
\lsp w, P_{\rm tot} w \rsp = \sum_{p \in \kappa \mathbb Z^2 }p\, \lsp (V w)_p, (V w)_p \rsp .
\end{align}
The map $V$ is called Lee--Low--Pines transformation \cite{LeeLowPines} and its inverse is given by $V^*(w_p) = e^{-i P_{\rm f}y}  (\varphi_p \otimes w_p ) $. 

From this definition it is not difficult to check that $\phi(\lambda)$ in \eqref{eq: phi(lambda) for all lambda} transforms into $
V \phi(\lambda) V^*  = \sum_{p \in \kappa  \mathbb Z^2} \phi_p(\lambda)$ where
\begin{align}\label{eq: momentum fiber of phi}
\phi_p(\lambda)  =  G(p-P_{\rm f},T-\lambda) - H_p(\lambda)  - X_p(\lambda) + P_p(\lambda)
\end{align}
is defined as an operator on $\HH_{N(\mu)-1}$ with 
\begin{align}
H_p(\lambda) & =  a(\eta)\frac{1}{\frac{1}{M} (p-P_{\rm f})^2+T-\lambda}a^*(\eta),\\[2mm]
X_p(\lambda) & = a(\eta) A^*_p(\lambda) + A_p(\lambda) a^*(\eta),\\[3mm]
P_p(\lambda) & =  \frac{1}{L^2} \sum_{k^2,l^2 > \mu} a_l^* \frac{1}{\frac{1}{M} (p-P_{\rm f} - k - l )^2 + T + k^2 + l^2 -\lambda} a_k,
\end{align}
and 
\begin{align}
\quad a(\eta) = \frac{1}{L} \sum_{k^2\le \mu} a_k , \quad A_p(\lambda) & = \frac{1}{L} \sum_{k^2 > \mu} \frac{1}{\frac{1}{M} (p-P_{\rm f}-k)^2 + T + k^2 -\lambda}a_k. \label{def: A(lambda)}
\end{align}
Note that the first summand in $\phi_p(\lambda)$ defines an unbounded operator whereas the three other terms can be shown to be bounded operators. The domain of essential self-adjointness of $\phi_p(\lambda)$ is the dense subspace consisting of all finite linear combinations of states of the form $\varphi_{k_1}\wedge \ldots \wedge \varphi_{k_{N(\mu)-1}}$ with $k_1,...,k_{N(\mu)-1}\in \kappa \mathbb Z^2$.

%---------------------------------------------------------------------------------------------------------------------------------------------------------------
\section{Localization of the polaron energy\label{sec: localization lemma}}
%---------------------------------------------------------------------------------------------------------------------------------------------------------------

By proposition \ref{prop: Birman-Schwinger} the lower bound $E(\mu,E_B)\ge \lambda$ is equivalent to $\phi(\lambda) \ge 0$. The next four sections are therefore devoted to the analysis of the condition $\phi_p(\lambda)\ge 0$ for the operator \eqref{eq: momentum fiber of phi} with $p\in \kappa \mathbb Z^2$. In view of the upper bound \eqref{sec: upper bound} it is sufficient to consider $\lambda \le E_0(\mu) +e_{\rm P}(\mu,E_B)$ from now on.

To prepare our first main statement we need to introduce a suitable orthogonal projector in the Hilbert space $\HH_{N(\mu)-1}$. For its definition let us give names to the subsets of the momentum lattice $ \kappa \mathbb Z^2$ that correspond to hole and particle momenta w.r.t.\ the Fermi sea,
\begin{align}
\Lambda_{\rm h}  = \big\{ k \in \kappa \mathbb Z^2  :\, k^2 \le \mu \big\}, \quad \quad \Lambda_{\rm p}  = \big\{ k \in  \kappa \mathbb Z^2  \, :\, k^2 > \mu\big\}.
\end{align}
Moreover for $\eps >0$ we set
\begin{align}
\Lambda_{{\rm p},\varepsilon}^\le =  \left\{ k  \in \Lambda_{\rm p} \, :\, \mu< k^2 \le  \left( 1 + \frac{1}{\varepsilon  \log \widetilde \mu} \right) \mu \right\}
\end{align}
and define the orthogonal projectors $\Pi_\eps$ and $\Pi^\perp_\eps =I-\Pi_\eps$ through
\begin{align} \label{eq: def of R of Pi}
& \text{Ran} ( \Pi_\eps )  =  \overline{\text{lin}} \big\{  a^*_{l_1}...a_{l_{m-1}}^*a_{k_1}...a_{k_m}\ket{ \rm {FS}_\mu} \,: \,  m\ge 1,\ k_1,...,k_m\in \Lambda_{\rm h},\  l_1,....,l_{m-1}\in \Lambda_{{\rm p},\varepsilon}^\le  \big\}.
\end{align}
Here $\overline{\text{lin} }V$ stands for the closure in $\HH_{N(\mu)-1}$ of the linear hull of the subset $V\subseteq \HH_{N(\mu)-1}$. For a better understanding of $\text{Ran} ( \Pi_\eps )$ and its orthogonal complement, let us recall that the set of all anti-symmetric products of $N(\mu)-1$ plane waves,
\begin{align}\label{eq: definition of total set D}
  D =   \big\{  a^*_{l_1}...a_{l_{m-1}}^*a_{k_1}...a_{k_m} \ket{ \rm {FS}_\mu} \, : \,  m\ge 1,\ k_1,...,k_m\in \Lambda_{\rm h},\  l_1,...,l_{m-1}\in \Lambda_{\rm p}  \big\}  ,
\end{align}
is a total set of the Hilbert space $\HH_{N(\mu)-1}$, i.e. $\overline{\text{lin} }D = \HH_{N(\mu)-1} $. A comparison with \eqref{eq: def of R of Pi} shows that $\Pi_\eps $ projects on all states in $\HH_{N(\mu)-1}$ that have particle modes occupied solely in the momentum lattice region $\Lambda_{{\rm p},\eps }^{\le}$ (this includes all states with zero particle modes occupied), whereas the range of $\Pi^\perp_\eps$ consists of states that have at least one mode occupied in $\Lambda_{{\rm p},\eps}^> =\Lambda_{\rm p}  \setminus \Lambda_{{\rm p},\eps}^{\le}$. 

In the next proposition we provide a lower bound for $\phi_p(\lambda)$ in terms of two operators that act only on $\text{Ran}(\Pi_\eps)$ and $\text{Ran}(\Pi^\perp_\eps)$, respectively. The physical meaning of the two subspaces is the following: On $\text{Ran}(\Pi^\perp_\eps)$ it is not clear how to obtain a suitable $L$-independent bound for the operator $P_p(\lambda)$ which is one of the main obstacles in the analysis. On this subspace we estimate the negative part of $P_p(\lambda)$ in terms of $G(p-P_{\rm f},T-\lambda)$. This is closely connected to the problem of obtaining a lower bound of $H$ uniformly in the system size $L\to \infty $. Such a bound, though necessary for the proof of Theorem \ref{thm: main theorem}, is however not much related to the asymptotic form of $E(\mu,E_B)-E_0(\mu)$. The latter will be determined on $\text{Ran}(\Pi_\eps)$ on which the operator $P_p(\lambda)$ is easily estimated with a suitable uniform bound. Hence on this subspace all of $G(p-P_{\rm f},T-\lambda)$ is available (and needed) for the analysis of the correct energy asymptotics. This explains the motivation behind the following decomposition of $\phi_p(\lambda)$. Since the operator $G( p - P_{\rm f},T-\lambda)$ is needed on both subspaces separately, it is an important step in our argument.
 
\begin{proposition}\label{lem: decomposing phi} There are constants $c_0 , \eps_0 >0$ such that for all $p\in \kappa \mathbb Z^2$, $L^2\vert E_B\vert \ge 1$, $ \mu / \vert E_B \vert \ge c_0 $ and $\eps \in (0,\eps_0)$, it holds that $\phi_p (\lambda) \ge \Phi_p(\lambda, \eps ) + \Psi_p(\lambda,\eps)$, with
\begin{align}
\Phi_p (\lambda,\eps ) \, & = \, \Pi_\eps \big(  G(p-P_{\rm f},T-\lambda) \,  - \,  H_p(\lambda) - X_p (\lambda)  - \varepsilon^{-1}  \big) \Pi_\eps,  \label{def: Phi of lambda}\\[2mm]
\Psi_p (\lambda,\eps ) \, & = \,  \Pi^\perp_\eps \big(  G(p-P_{\rm f},T-\lambda) \, + \, P_p (\lambda) \,  - K(\eps ,\widetilde \mu)  \big) \Pi^\perp_\eps ,\label{def: Psi of lambda}
\end{align}
and $K(\eps ,\widetilde \mu) = \eps^{-1/2}(\eps^{-1/2} + \sqrt {\log\widetilde \mu} + \eps \log\widetilde \mu)$. We use the notation $\widetilde \mu = \mu/ \vert E_B \vert$.
\end{proposition}

\noindent \textbf{Remark 3.1.} From the discussion above it is clear that the $\eps$- and $\widetilde \mu$-dependent errors in \eqref{def: Phi of lambda} and \eqref{def: Psi of lambda} have physically different meanings. Eventually only the error in \eqref{def: Phi of lambda} enters the constant on the right side of \eqref{eq: main estimate}. For that reason, we do not optimize the error terms as $\widetilde \mu \to \infty$, and always consider $\eps$ sufficiently small but fixed w.r.t. $\widetilde \mu$ and $L>0$.\medskip

Before we come to the proof of the proposition, we state two helpful results about the asymptotics of $G(q,\tau)$ and $e_{\rm P}(\mu,E_B)$.

%---------------------------------------------------------------------------------------------------------------------------------------------------------------
\subsection{Asymptotics of $G(q,\tau)$ and $e_{\rm P}(\mu,E_B)$}
\label{sec: Asymptotics of G and $e_P$}
%---------------------------------------------------------------------------------------------------------------------------------------------------------------

As the following bound will be used several times, we note that it follows easily with the aid of Lemma \ref{lemma: replacing sums by integrals},
\begin{align}\label{eq: general sum for k^2 between a mu and b mu}
\bigg\vert \frac{1}{L^2}\sum_{a\mu \le k^2  < b \mu } 1 - \frac{(b-a)\mu}{4\pi} \bigg\vert \le  \frac{2}{\pi L} \big( \sqrt{a\mu} + \sqrt{b\mu} \big) + \frac{6}{L^2}
\end{align}
for any $b>a\ge 0$.

The first lemma of this section tells us the error for replacing the sum in $G(q,\tau )$ by the corresponding integral which can be evaluated explicitly.

\begin{lemma}\label{lem: asymptotic form of G} There are constant $c_0,  C>0$ such that 
\begin{align}\label{eq: asymptotic form of G}
\left| G(q,\tau) - \frac{1}{4\pi m }\, \log\Bigg( \frac{  \frac{q^2}{M+1} + m  \mu + \tau }{ \vert E_B \vert } \Bigg) \right| \le C\bigg( 1 + \frac{\mu }{( \mu + \tau ) \log ( \mu/ \vert E_B\vert )}\bigg)^3 
\end{align}
for all $q\in \mathbb R^2$, $\tau > -\mu$, $L^2\vert E_B\vert \ge 1$ and $\mu / \vert E_B \vert  \ge c_0$. Recall $m = 1+1/M$.
\end{lemma}

The proof of the lemma is postponed to Section \ref{sec: asymptotic form of G}. Utilizing this lemma, we can derive an asymptotic formula for $e_{\rm P}(\mu,E_B)$ as $ \mu/ \vert E_B\vert  \to \infty$. The precise statement is

\begin{lemma}\label{lemma: asymptotic polaron energy} There are constants $c_0,C>0$ such that the polaron energy satisfies
\begin{align}\label{eq: assymptotic formula for polaron energy}
\left\vert  e_{\rm P}(\mu,E_B) +  \big(1+\frac{1}{M}\big)\frac{ \mu}{\log( \mu/ \vert E_B\vert ) }   \right\vert \, \le \, C\,  \frac{ \mu } { ( \log ( \mu/ \vert E_B\vert ) ) ^2}
\end{align}
for all $L^2\vert E_B\vert \ge 1$ and $\mu / \vert E_B \vert \ge c_0$.
\end{lemma}
\begin{proof} Let us set $z_{\rm P} = |e_{\rm P}(\mu,E_B)|$ and $\widetilde \mu = \mu/ \vert E_B\vert$. To prove suitable upper and lower bounds for $z_{\rm P}$ we first show $z_{\rm P}\le \mu$ for all $\mu/ \vert E_B\vert \ge c_0$ given that the constant $c_0$ is chosen large enough.

Consider the set of parameters for which $z_{\rm P}$ exceeds the value $\mu$,
\begin{align}
\mathcal M_{c_0} = \big\{ (L, \mu, E_B)  \, : z_{\rm P} > \mu \ge c_0 \vert E_B\vert \, \, \text{and} \, \, L^2\vert E_B\vert \ge 1 \big\} \subseteq  \mathbb R_+ \times \mathbb R_+ \times \mathbb R_- .
\end{align}
By monotonicity of $G(q, \tau)$ in the $\tau$ variable, we have
$G(k,-k^2+z_{\rm p}) \ge G(k,0)$ for all $k^2 \le \mu$ and $(L,\mu,E_B) \in \mathcal M_{c_0}$. By Lemma \ref{lem: asymptotic form of G} this implies
\begin{align}
G(k,0 ) \ge \frac{1}{4\pi m }\log (m c_0 ) - C \bigg( 1 + \frac{1}{\log c_0}\bigg)^3 \ge \frac{1}{8\pi m } \log (m c_0).
\end{align}
Inserting this into the polaron equation \eqref{eq: polaron equation} and employing \eqref{eq: general sum for k^2 between a mu and b mu} leads to
\begin{align}
z_{\rm P} \le  \frac{2m \mu }{\log (m c_0) }  \bigg( 1 +  \frac{8}{\sqrt {c_0 }} + \frac{24\pi }{c_0} \bigg),
\end{align}
which implies $z_{\rm P}\le \mu$ for $c_0$ large enough. Hence $\mathcal M_{c_0}$ is empty and we can assume $z_{\rm P}\le \mu$.

Utilizing again monotonicity of $G(q,\tau)$, we get $G(k,-k^2+z_{\rm P}) \le G(k,\mu )$. By \eqref{eq: asymptotic form of G} we have for all $k^2\le \mu$,
\begin{align}
G(k,\mu ) \le \frac{1}{4\pi m }\, \log\Bigg( \frac{  \frac{\mu }{M+1} + m  \mu + \mu }{ \vert E_B \vert } \Bigg) + C_1 \le \frac{1}{4\pi m }\, \log \widetilde \mu + C_2
\end{align}
for two constants $C_1,C_2>0$. Using \eqref{eq: polaron equation} together with \eqref{eq: general sum for k^2 between a mu and b mu}, we obtain the lower bound
\begin{align}\label{eq: lower bound for z_P}
z_{\rm P} \ge  \frac{\mu }{ m^{-1} \log \widetilde \mu + C_2 } \bigg( 1 - \frac{C_3}{\sqrt{\widetilde \mu} }	\bigg) \ge \frac{m \mu }{  \log \widetilde \mu } - C / ( \log \widetilde \mu)^2  .
\end{align}

With the lower bound \eqref{eq: lower bound for z_P} we can estimate for $0\le k^2 \le \mu$,
\begin{align}
G (k, -k^2 + z_{\rm P}) & \ge \frac{1}{4\pi m } \log ( m \widetilde \mu -\widetilde \mu )   - C_1 \bigg( 1 + \frac{\mu }{z_{\rm P} \log \widetilde \mu }\bigg)^3  \ge \frac{1}{4\pi m } \log \widetilde \mu    - C_2.
\end{align}
Similarly as above, using the polaron equation and \eqref{eq: general sum for k^2 between a mu and b mu}, one finds
\begin{align}
z_{\rm P} \le   \frac{m \mu}{\log \widetilde \mu} \bigg(  \frac{1}{1-C_2  /  \log \widetilde \mu }\bigg) \bigg( 1 + \frac{C_3}{\sqrt{\widetilde \mu} }	\bigg),
\end{align}
which gives the desired upper bound.
\end{proof}

\noindent \textbf{Remark 3.2.} For $\lambda  \le E_0(\mu) + e_{\rm P}(\mu,E_B)$ it follows from $T\restriction \HH_{N(\mu)-1} \ge E_0(\mu)-\mu$ that there are constants $c_0,C>0$ such that 
\begin{align}\label{eq: operator bounds for G}
\pm \Bigg( G(p-P_{\rm f}, T-\lambda) - \frac{1}{4\pi m}   \, \log\Bigg( \frac{  \frac{(p-P_{\rm f})^2}{M+1} + m  \mu + T-\lambda }{ \vert E_B \vert } \Bigg)  \Bigg) \le C 
\end{align}
as operator inequalities on $\HH_{N(\mu)-1}$ for all $L^2 \vert E_B \vert \ge 1$ and $\mu / \vert E_B \vert \ge c_0$. A useful implication of this bound is
\begin{align}\label{eq: operator bounds for G 2}
 G(p-P_{\rm f}, T-\lambda) \restriction \HH_{N(\mu)-1} \ge  \frac{1}{4\pi m} \log (\mu/\vert E_B \vert ) - C.
\end{align}
 
%---------------------------------------------------------------------------------------------------------------------------------------------------------------
%---------------------------------------------------------------------------------------------------------------------------------------------------------------
\subsection{Proof of Proposition \ref{lem: decomposing phi}}
%---------------------------------------------------------------------------------------------------------------------------------------------------------------
%---------------------------------------------------------------------------------------------------------------------------------------------------------------

Since $G(p-P_{\rm f},T-\lambda)$ and $H_p(\lambda)$ both commute with the projector $\Pi_\eps $, we have
\begin{align}
\phi_p(\lambda)\, & =  \, \Pi_\eps  \, \phi_p (\lambda) \, \Pi_\eps     \, + \,\Pi^\perp_\eps  \, \phi_p (\lambda) \,  \Pi^\perp_\eps + \Big( \Pi_\eps  \big(-\, X_p (\lambda)  \, +\, P_p (\lambda) \big) \Pi^\perp_\eps \, + \text{h.c.}\Big)  .
\end{align}
The statement of Proposition \ref{lem: decomposing phi} is a consequence of the following estimates. Note that for notational convenience we estimate the constant $C$ from above by $\varepsilon^{-1/2}$.

\begin{lemma} There are constants $c_0,\eps_0, C>0$ such that for all $p\in \kappa \mathbb Z^2$, $L^2\vert E_B\vert \ge 1$, $\widetilde \mu = \mu / \vert E_B\vert \ge c_0$ and  $\varepsilon\in (0,\eps_0)$,
\begin{align}
\Pi_\eps \, P_p(\lambda) \Pi_\eps \,&  \ge \, -\frac{C}{\sqrt \eps}\, \Pi_\eps,\label{bound: Pi P Pi} \\[1mm]
\Pi^\perp_\eps  \, H_p(\lambda) \,  \Pi^\perp_\eps \, &  \le \,  C \varepsilon \log \widetilde \mu \ \Pi^\perp_\eps  ,\label{bound: Pi_perp H Pi_perp}\\[2.5mm]
\Pi^\perp_\eps \, X_p(\lambda) \, \Pi^\perp_\eps   \, & \le \,  C \sqrt{\log\widetilde \mu} \ \Pi^\perp_\eps  ,\label{bound: Pi_perp X Pi_perp}
\end{align}
and
\begin{align}
\Pi_\eps \, X_p(\lambda) \, \Pi^\perp_\eps   \, + \, \textnormal{h.c.} \,&  \le \, C \eps^{1/2} \log \widetilde \mu \ \Pi^\perp_\eps \, + \, C \eps^{-1/2} \ \Pi_\eps ,\label{bound: Pi X Pi_perp} \\[1mm]
\Pi_\eps \, P_p(\lambda) \, \Pi^\perp_\eps \, + \, \textnormal{h.c.} \, & \ge \, - \frac{C}{\sqrt \eps}.\label{eq: bound Pi P Pi_perp }
\end{align}
\end{lemma}

\begin{proof}
\textbf{Line  \eqref{bound: Pi P Pi}}. Using $a_k \Pi_\eps = 0$ for all $k \in \Lambda_{\rm p} \setminus \Lambda_{{\rm p},\eps}^{\le }$ together with the pull-through formula \eqref{eq: pull through formula}, we obtain\allowdisplaybreaks
\begin{align}
\Pi_\eps \, P_p(\lambda)\, \Pi_\eps  
% = \Pi_\eps  \Bigg(\frac{1}{L^2}\sum_{k,l \in \Lambda_{{\rm p},\eps}^{\le}} a_l^*\frac{1}{\frac{1}{M} (p-P_f+l+k)^2+T+l^2+k^2-\lambda} a_k  \Bigg) \Pi_\eps \nonumber\\
 & = \Pi_\eps \Bigg(\frac{1}{L^2}\sum_{k,l \in \Lambda_{{\rm p},\eps }^{\le}} a_l^* a_k \frac{1}{\frac{1}{M} (p-P_{\rm f}  - l)^2+T+l^2-\lambda} \Bigg) \Pi_\eps. \label{eq: bound for Pi P Pi line 1}
\end{align}
By means of the commutation relations $a_k a_l^* +  a_l^*  a_k = \delta_{kl} $ and
\begin{align}
& \frac{1}{L^2}\sum_{l  \in \Lambda_{{\rm p},\eps}^{\le}} \frac{1}{\frac{1}{M}(p-P_{\rm f} - l )^2 + T +  l^2-\lambda}  \, \ge\, 0
\end{align}
on $\HH_{N(\mu)-1}$, we further get
\begin{align}
\eqref{eq: bound for Pi P Pi line 1}   \ge \Pi_\eps \Bigg(-\frac{1}{L^2}\sum_{k,l \in \Lambda_{{\rm p},\eps}^{\le}} a_k\frac{1}{ \frac{1}{M} (p-P_{\rm f} )^2+T-\lambda}  a_l^*   \Bigg) \Pi_\eps.
\end{align}
For $f(\mu,E_B)>0$ we proceed by estimating the right side from below in terms of the operator
\begin{align}
-\frac{f(\mu,E_B)}{2}  \Bigg( \frac{1}{L^2}\sum_{k,l \in \Lambda_{{\rm p},\eps }^{\le}} a_k a_l^* \Bigg)  - \frac{1}{2f(\mu,E_B)}  \Bigg( \frac{1}{L^2} \sum_{k,l \in \Lambda_{{\rm p},\eps}^{\le}}  a_k \frac{1}{( \frac{1}{M} (p-P_{\rm f})^2 + T-\lambda)^2}  a_l^* \Bigg) \label{eq: bound for Pi P Pi line 2}
\end{align}
acting on $\Pi_\eps \HH_{N(\mu)-1}$. In the first summand, we use the CAR together with \eqref{eq: general sum for k^2 between a mu and b mu} for $a= \mu$, $b = \mu (1+ \frac{1}{\eps  \log \widetilde \mu } )$, and further employ $ L^2 \vert E_B\vert \ge 1$ and $\widetilde \mu\ge c_0$. This gives
\begin{align}\label{eq: bound for the summ over a star a}
\frac{1}{L^2}\sum_{ k,l \in \Lambda_{p,\varepsilon}^{\le}  } a_k a_l^*\le 
\frac{1}{L^2}\sum_{ k \in \Lambda_{p,\varepsilon}^{\le}  } 1 \le   \frac{C \mu}{  \varepsilon \log \widetilde \mu } .
\end{align}
In the second summand in \eqref{eq: bound for Pi P Pi line 2}, we use $T-\lambda > 0$ on $\HH_{N(\mu)}$ in order to neglect the positive operator $\frac{1}{M} (p- P_{\rm f})^2$ in the denominator. Then we use again the pull-through formula and the CAR to obtain
\begin{align} 
   \frac{1}{L^2} \sum_{k,l\in \Lambda^\le_{{\rm p},\eps }}   a_ k \frac{1}{(\frac{1}{M}  (p-P_{\rm f})^2 + T -\lambda)^2}a_l^* & \le  \frac{1}{L^2} \sum_{k,l\in \Lambda^\le_{{\rm p},\eps }}   a_ k \frac{1}{(T -\lambda)^2}a_l^*\notag \\
 & \le \frac{1}{L^2} \sum_{k \in \Lambda^\le_{{\rm p},\eps }}    \frac{1}{(T+k^2-\lambda)^2 }.\label{eq: bound for Pi P Pi line 3}
\end{align}
Note that in the last step, we applied the inequality
\begin{align}\label{eq: positivity of a_k a_l^* term}
  \frac{1}{L^2} \sum_{k,l\in \Lambda^\le_{{\rm p},\eps }}  a_l^*  \frac{1}{(T+k^2+l^2-\lambda)^2} a_k \ge 0,
\end{align}
which is verified by writing 
\begin{align}
 \frac{1}{(T+k^2+l^2-\lambda)^2}  \, = \, \int_0^\infty \exp \left( -t (T-\lambda + k^2 +l^2)^2\right) \, \D t
\end{align}
and estimating
\begin{align}
\exp \left( -t (T-\lambda + k^2 +l^2)^2\right)  \ge \exp\left( -t 4 l^4 \right)  \exp \left( -t 2(T-\lambda )^2 \right) \exp\left( -t 4k^4 \right)  .
\end{align}
With $T\ge E_0(\mu)  -\mu $ on $\HH_{N(\mu)-1}$ and $E_0(\mu) - \lambda \ge  \vert e_P(\mu,E_B)\vert$ we next get
\begin{align} 
\eqref{eq: bound for Pi P Pi line 3}  \le \frac{1}{L^2} \sum_{k^2 > \mu  }    \frac{1}{(k^2-\mu +  \vert e_{\rm P}(\mu,E_B)\vert )^2 }.
\end{align}
By Lemma \ref{lemma: replacing sums by integrals} and the estimate
\begin{align}
\int\limits_{\sqrt\mu}^\infty \frac{\D t}{(t^2 - \mu + |e_{\rm P}(\mu,E_B)|)^2} & \, \le  \, \int\limits_{\sqrt\mu}^\infty \frac{\D t}{((t - \sqrt\mu)^2  + |e_{\rm P}(\mu,E_B)|)^2} \,= \, \frac{\pi }{4 |e_{\rm P}(\mu,E_B)|^{3/2}},
\end{align}
one finds the upper bound 
\begin{align}
\eqref{eq: bound for Pi P Pi line 3} & \le  \frac{1}{4\pi \vert e_{\rm P}(\mu,E_B)\vert } + \frac{1}{L \vert e_{\rm P}(\mu,E_B)\vert ^{3/2}} + \bigg( \frac{4\sqrt{\mu}} {L} + \frac{6}{L^2} \bigg) \frac{1}{\vert e_{\rm P}(\mu,E_B)\vert ^2} \le  \frac{C \log\widetilde \mu}{\mu}.
\end{align}
Hence,
\begin{align}
& \Pi_\eps  \, P(\lambda)  \, \Pi_\eps  \ge -  C \left( \frac{f(\mu,E_B) \mu }{ \eps \log \widetilde \mu}  + \frac{\log\widetilde \mu}{f(\mu,E_B) \mu}\right)   \Pi_\eps   \ge  -  \frac{C}{\sqrt \eps } \, \Pi_\eps ,
\end{align}
if we choose $f(\mu,E_B) = (\sqrt {\eps } \log\widetilde \mu ) / \mu $.\medskip

\noindent \textbf{Line \eqref{bound: Pi_perp H Pi_perp}}. Since states of the form $\widetilde w = a^*(\eta) w \in \HH_{N(\mu)}$ with $ w  \in \text{Ran}(\Pi^\perp_\eps)$ have at least one momentum mode occupied in $\Lambda_{{\rm p},\eps }^{\ge}$, it follows that
\begin{align}
\Pi^\perp_\eps  \, H_p( \lambda) \,  \Pi^\perp_\eps \, \le \, \Pi^\perp_\eps \left( \frac{a(\eta )a^*(\eta )}{
\vert e_{\rm P}(\mu,E_B) \vert +  \mu / (\eps \log \widetilde \mu )} \right)  \Pi^\perp_\eps .
\end{align}
The remaining expression is estimated using 
\begin{align}\label{eq: estimate eta mu}
a(\eta )a^*(\eta ) \le \frac{1}{L^2} \sum_{ k^2\le \mu }1\le C \mu  
\end{align}
which follows from the CAR in combination with \eqref{eq: general sum for k^2 between a mu and b mu}.\medskip

\noindent \textbf{Lines \eqref{bound: Pi_perp X Pi_perp} and \eqref{bound: Pi X Pi_perp}}. It is straightforward to verify that for any two orthogonal projectors $Q,\widetilde Q$ acting on $\HH_{N(\mu)-1}$ and for any $f(\mu,E_B) , g(\mu,E_B)>0 $,
\begin{align}
 Q X_p(\lambda) \widetilde Q\, + \, \text{h.c.} & \, \le \,  Q \left( f(\mu,E_B)  \, A_p (\lambda) A_p^*(\lambda)  + \frac{a(\eta) a^*(\eta) }{g(\mu,E_B)}\right)  Q \, \nonumber\\[0.5mm] 
& \hspace{1.5cm} + \, \widetilde Q \left(  g(\mu,E_B)  \, A_p(\lambda) A_p^*(\lambda) + \frac{a(\eta) a^*(\eta)}{f(\mu,E_B)}   \right) \widetilde Q. \label{eq: Q X Q tilde line 1}
\end{align}
Similar as in the analysis of \eqref{eq: bound for Pi P Pi line 1}, one further shows
\begin{align}
A_p(\lambda) A^*_p(\lambda) \restriction \HH_{N(\mu)-1} \, \le \,  C\frac{\log \widetilde \mu}{\mu}.
\end{align} 
Together with \eqref{eq: estimate eta mu} and \eqref{eq: Q X Q tilde line 1} this leads to
\begin{align}
& Q \, X_p(\lambda)\,  \widetilde Q\, + \, \text{h.c.}\, \notag\\
& \quad \le \, C\left( \frac{f(\mu,E_B)\log \widetilde \mu}{\mu} + \frac{\mu}{g(\mu,E_B)} \right)   Q \, + \, C \left(\frac{g(\mu,E_B)\log \widetilde \mu}{\mu} + \frac{\mu}{f(\mu,E_B)}  \right)    \widetilde Q \label{eq: Q X Q tilde line 2}.
\end{align}
For $f(\mu,E_B) = g(\mu,E_B) =  \mu/\sqrt{\log\widetilde \mu}$, this shows \eqref{bound: Pi_perp X Pi_perp}, whereas the inequality in \eqref{bound: Pi X Pi_perp} follows from $f(\mu,E_B) = \sqrt{\eps} \mu $ and $g(\mu,E_B) =  \mu /(\sqrt{\eps} \log\widetilde \mu)$. (In the latter case we set $Q = \Pi^\perp_\eps $ and $\widetilde Q=\Pi_\eps$.)\medskip

\noindent \textbf{Line \eqref{eq: bound Pi P Pi_perp }}. Using $a_l \Pi_\eps = 0$ for $l \in \Lambda_{\rm p}\setminus \Lambda_{{\rm p},\eps}^{\le}$ and the pull-through formula together with the CAR, we find 
\begin{align}
  \Pi_\eps \, P_p(\lambda) \, \Pi^\perp_\eps   
 = \Pi _\eps \Bigg(-\frac{1}{L^2}\sum_{\substack{ l \in \Lambda^\le_{{\rm p},\eps } \\ k^2 > \mu   }} a_k\frac{1}{\frac{1}{M}(p- P_{\rm f})^2+T-\lambda}  a_l^*   \Bigg) \Pi^\perp_\eps  , \label{eq: Pi P Pi_perp line 1}
\end{align}
where we made use of 
\begin{align}
\Pi_\eps  \frac{1}{\frac{1}{M}(p-P_{\rm f} - l)^2+T+l^2-\lambda  }\Pi^\perp_\eps =  \frac{1}{\frac{1}{M}(p-P_{\rm f} - l)^2+T+l^2-\lambda  } \Pi_\eps \Pi^\perp_\eps  = 0.
\end{align}
From here the proof works the same way as for \eqref{eq: bound for Pi P Pi line 2} (with the difference that we end up with an identity on the right side). We obtain
\begin{align}
\Pi_\eps \, P_p(\lambda) \, \Pi^\perp_\eps + \text{h.c.}  \ge  -  \frac{C}{\sqrt \eps },
\end{align}
which completes the proof of the lemma and thus also the proof of Lemma \ref{lem: decomposing phi}.
\end{proof}

%---------------------------------------------------------------------------------------------------------------------------------------------------------------
%---------------------------------------------------------------------------------------------------------------------------------------------------------------
\section{Analysis of $\Phi_p(\lambda,\eps)$: perturbed polaron equation \label{sec: analysis of Phi}}
%---------------------------------------------------------------------------------------------------------------------------------------------------------------
%---------------------------------------------------------------------------------------------------------------------------------------------------------------

In this section we show that the condition $\Phi_p(\lambda,\eps)\ge 0 $ leads to a perturbed polaron equation for $\lambda$ and then provide a suitable estimate for the solution of this equation.

In order to obtain the presumably optimal asymptotics of the error in \eqref{eq: main estimate}, we introduce another orthogonal projector $\Pi_{\eps,1}$ with $\text{Ran}(\Pi_{\eps,1} ) \subseteq \text{Ran}(\Pi_\eps)$ defined as the closed subspace of all states containing exactly one unoccupied momentum mode (a hole) in the lattice region
\begin{align}
\Lambda_{ {\rm h} , \eps }^\le  = \left\{ k \in \kappa \mathbb Z^2 \, : \, k^2 \le  \left( 1 - \frac{ 1 }{ \eps \log \widetilde \mu} \right) \mu \right\} \subset \Lambda_{\rm h}.
\end{align}
More precisely we set for $n\ge 0$, $\text{Ran}(\Pi_{\eps , n } ) =  \overline{\text{lin}}  V_{\eps,n} $ with  $V_{\eps ,n } \subset \HH_{N(\mu)-1}$ the subset
\begin{align}
V_{\eps,n}  =   \Big\{  w = a^*_{l_1}...a_{l_{m-1}}^*a_{k_1}...a_{k_m} \ket{\rm {FS}_\mu}  \big|\, &  m\ge 1,\ k_1,...,k_m\in \Lambda_{\rm h},\  l_1,....,l_{m-1}\in \Lambda_{{\rm p},\eps}^\le ,\notag \\
&\hspace{2cm} \text{and}\ \sum_{k\in \Lambda_{ {\rm h} , \eps }^\le } a_k a_k^* w  \, = \, n w \Big\}.
\end{align}
Clearly $\text{Ran}(\Pi_\eps ) = \bigoplus_{n\ge 0}  \text{Ran}(\Pi_{\eps,n})$. (Note that the operator $\sum_{k\in \Lambda_{ {\rm h} , \eps }^\le } a_k a_k^*$ counts the number of holes in $\Lambda_{ {\rm h} , \eps }^\le $.)

The next lemma provides a more accurate localization of the polaron energy inside the subspace $\text{Ran}(\Pi_{\eps,1})$. 

\begin{lemma} There are constants $ c_0,\varepsilon_0>0$ such that for all $p\in \kappa \mathbb Z^2$, $L^2\vert E_B\vert \ge 1$, $\mu / \vert E_B\vert \ge c_0$ and  $\varepsilon\in (0,\varepsilon_0)$, we have
\begin{align}
\Phi_p(\lambda, \varepsilon) \, & \ge \, \Pi_{\eps,1}  \big( pol(\lambda) - \varepsilon^{-3} \big) \Pi_{\eps,1} 
\end{align}
where $pol(\lambda)= G(0,T-\lambda) - a(\eta) (T-\lambda)^{-1} a^*(\eta)$.
\end{lemma}
\begin{proof} We write $\Pi_\eps = \Pi_{\eps,0}  + \Pi_{\eps,1} + \Pi_{\eps,2+ } $ with $\Pi_{\eps,2+ } = \sum_{n\ge 2} \Pi_{\eps,n}$. Below we prove the inequality
\begin{align}
\Phi_p(\lambda, \varepsilon) \, & \ge \, \Pi_{\eps,1}  \big( G(p-P_{\rm f},T-\lambda) - H_p(\lambda)  - C  \varepsilon^{-2 }  \big) \Pi_{\eps,1}  \nonumber \\[2mm] & 
\hspace{0.3cm} \, +\,  (\Pi_{\eps,0}  + \Pi_{\eps,2+} ) \big(G(p-P_{\rm f},T-\lambda) - C \varepsilon \log\widetilde \mu \big) (\Pi_{\eps,0} + \Pi_{\eps,2+} )\label{eq: first bound for localizing pol}
\end{align}
from which the statement of the lemma follows by
\begin{align}\label{eq: G minus H bound}
G(p-P_{\rm f},T-\lambda) - H_p (\lambda) \ge G(0,T-\lambda) - a(\eta) (T-\lambda)^{-1} a^*(\eta) - C
\end{align}
together with inequality \eqref{eq: operator bounds for G 2}. By choosing $\varepsilon_0$ small enough the second line in \eqref{eq: first bound for localizing pol} is positive for all $\widetilde \mu \ge c_0$. The bound in \eqref{eq: G minus H bound} is a direct consequence of \eqref{eq: operator bounds for G} and the fact that $T-\lambda \ge \vert e_{\rm P}(\mu,E_B)\vert$ on $\HH_{N(\mu)}$.

The derivation of \eqref{eq: first bound for localizing pol} occupies the remainder of this proof. To this end, note
\begin{align}
\Pi_\eps G(p-P_{\rm f},T-\lambda)\Pi_\eps  &  = \, \Pi_{\eps,1} G(p-P_{\rm f},T-\lambda) \Pi_{\eps,1} \notag\\[1mm]
& \quad  + (\Pi_{\eps ,0 } +\Pi_{\eps, 2+ } ) G(p-P_{\rm f},T-\lambda) (\Pi_{\eps, 0} +\Pi_{\eps,2+}).
\end{align}
$\bullet$ Introducing $\Lambda_{{\rm h} ,\eps}^> = \Lambda_{\rm h} \setminus \Lambda_{{\rm h},\eps }^\le$ and $a(\eta^{>}_\eps) = \sum_{k\in \Lambda_{{\rm h},\eps}^>} a_k$, we can start with
\begin{align}
\Pi_{\eps,0}  \, H_p(\lambda)\, \Pi_{\eps,0}  \le \Pi_{\eps,0}  \left( \frac{ a(\eta^{>}_\eps)a^*(\eta^{>}_\eps)}{|e_{\rm P}(\mu,E_B)|}\right) \Pi_{\eps,0} \le C\eps^{-1} \, \Pi_{\eps,0}
\end{align}
which follows from 
\begin{align}\label{eq: T minus lambda on H_N}
\left(\frac{1}{M} ( p - P_{\rm f}) ^2 + T-\lambda\right)\rst \HH_{N(\mu)} \ge |e_{\rm P}(\mu,E_B)| \ge \frac{C \mu}{\log\widetilde\mu } ,
\end{align}
$a^*(\eta)\Pi_0(\varepsilon)   = a^*(\eta^>_\varepsilon)\Pi_0(\varepsilon)$, and
\begin{align} \label{eq: estimate eta geq}
a(\eta^>_\varepsilon )a^*(\eta^{>}_\varepsilon) \ \le \ \frac{1}{L^2} \sum_{k\in \Lambda_{h,\varepsilon}^>} 1\ \le \ \frac{C \mu}{\varepsilon \log\widetilde \mu}.
\end{align}
The latter is obtained via \eqref{eq: general sum for k^2 between a mu and b mu}.\medskip
%\begin{align}
%\Bigg \vert \frac{1}{L^2} \sum_{k\in \Lambda_h^>} 1 - \frac{\mu}{4\pi \varepsilon \log\widetilde \mu} \Bigg \vert  \le ...
%\end{align}

\noindent $\bullet$ Next we consider
\begin{align}
\Pi_{\eps,0}\, & H_p(\lambda)\, \Pi_{\eps,1}  \, +\, \text{h.c.}  \notag \\[0mm]
&  = \ \Pi_{\eps,0} \, a(\eta^{>}_\eps) \frac{1}{\frac{1}{M} (p-P_{\rm f})^2 + T - \lambda}a^*(\eta) \, \Pi_{\eps,1} \, +\, \text{h.c.}\notag \\[0mm]
&  \le \frac{( \varepsilon \log \widetilde \mu)^2}{\mu} \, \Pi_{\eps,0}\,  a(\eta^{>}_\varepsilon)a^*(\eta^{>}_\eps)\, \Pi_{\eps,0} + \frac{\mu}{(\varepsilon \log \widetilde \mu)^2}\, \Pi_{\eps,1}\,  a(\eta ) \frac{1}{( T - \lambda)^2}a^*(\eta)\ \Pi_{\eps,1}\notag \\[0mm]
&  \le C \big( \varepsilon \log\widetilde \mu \, \Pi_{\eps,0 } +  
\varepsilon^{-2} \, \Pi_{\eps,1} \big),
\end{align}
where we made use of \eqref{eq: estimate eta mu}, \eqref{eq: T minus lambda on H_N} and \eqref{eq: estimate eta geq}.\medskip

\noindent $\bullet$  The contribution
\begin{align}
\Pi_{\eps,0}  \, H_p(\lambda)\, \Pi_{\eps,2+} \, +\,  \text{h.c.} = 0
\end{align}
vanishes identically since 
\begin{align}
a(\eta) \frac{1}{\frac{1}{M} (p-P_{\rm f})^2 + T - \lambda} a^*(\eta) \Pi_{\eps,2+}  w  \in \text{Ran}(\Pi_{\eps,1} ) \oplus \text{Ran}(\Pi_{\eps,2+} ) 
\end{align}
for any $w\in \HH_{N(\mu)-1}$ and $\Pi_{\eps,0}\Pi_{\eps,1} = \Pi_{\eps,0}\Pi_{\eps,2+} =0 $. (The operator $a^*(\eta)$ can reduce the number of unoccupied modes at most by one.)\medskip

\noindent $\bullet$ We proceed with
\begin{align}\label{eq: Pi_1 H Pi_2+}
& \Pi_{\eps,1} \, H_p(\lambda)\, \Pi_{\eps,2+} \, +\, \text{h.c.} \, = \, \Pi_{\eps,1}  \, a(\eta^>_\eps) \frac{1}{\frac{1}{M}(p - P_{\rm f})^2 + T - \lambda}a^*(\eta)\, \Pi_{\eps,2+} \, +\, \text{h.c.}
\end{align}
which holds because of 
\begin{align}
( a (\eta) - a (\eta^>_\eps)) \frac{1}{\frac{1}{M}(p - P_{\rm f})^2 + T - \lambda}a^*(\eta)\, \Pi_{\eps,2+} w  \in \text{Ran}(\Pi_{\eps,2+} )
\end{align}
and $\Pi_{\eps,1} \Pi_{\eps,2+}  = 0$. (Note that $a (\eta) - a (\eta^>_\eps)$ adds an unoccupied mode in $\Lambda^{\le}_{{\rm h},\eps}$.) We estimate the r.h.s.\ of \eqref{eq: Pi_1 H Pi_2+} from above by
\begin{align}
&  \frac{\varepsilon \log \widetilde \mu}{\mu } \, \Pi_{\eps,1}  a(\eta^{>}_\varepsilon )a^*(\eta^{>}_\varepsilon )\, \Pi_{\eps,1}  + \frac{ \mu}{\varepsilon \log \widetilde \mu}\,  \Pi_{\eps, 2+} \,   a(\eta ) \frac{1}{(\frac{1}{M} ( p -P_{\rm f})^2 + T - \lambda)^2}a^*(\eta) \, \Pi_{\eps,2+} \nonumber \\[2mm]
& \hspace{8cm}\le C \big( \Pi_{\eps,1} + \eps  \log\widetilde \mu \, \Pi_{\eps, 2+} \big),
\end{align}
where we used another time that states of the form $\psi =  a^*(\eta)\Pi_{\eps,2+} w \in \HH_{N(\mu)}$ are either zero or have at least one unoccupied mode in $\Lambda_{h,\varepsilon}^\le$. The latter implies
\begin{align}\label{eq: bound for (T-lambda) to minus s}
\scp{\psi}{ (T - \lambda)^{-s}  \psi } \le \scp{\psi}{  \psi } \frac{(\eps \log \widetilde \mu )^s}{\mu^s}\quad (s>0).
\end{align}
\medskip

\noindent $\bullet$ In the bound for $\Pi_{\eps,2+}  \, H_p(\lambda)\, \Pi_{\eps,2+} $ we use \eqref{eq: bound for (T-lambda) to minus s} with $s=1$ to get
\begin{align} 
\Pi_{\eps,2+}\, H_p(\lambda) \, \Pi_{\eps,2+} \le
 \frac{C\eps  \log\widetilde \mu}{\mu} \Pi_{\eps,2+}  a(\eta) a^*(\eta) \Pi_{\eps, 2+}  \le \ C \eps \log\widetilde\mu \, \Pi_{\eps,2+}
\end{align}
by means of \eqref{eq: estimate eta mu}.

So far we have shown
\begin{align}
\Pi_\eps\, H_p(\lambda) \, \Pi_\eps - \Pi_{\eps,1} \, H_p(\lambda) \, \Pi_{\eps,1} \, \ge \, C \eps^{-2}\, \Pi_{\eps,1} + C \eps \log \widetilde \mu \, \big( \Pi_{\eps,0} + \Pi_{\eps,2+} \big) .
\end{align}
For the bounds involving $X_p(\lambda)$, we recall \eqref{eq: Q X Q tilde line 2}.\medskip

\noindent $\bullet$ With $f(\mu,E_B) = g(\mu,E_B) =  \mu/\sqrt{\log\widetilde \mu}$, we obtain 
\begin{align}
(\Pi_{\eps,0}  +\Pi_{\eps,2+} ) \, X_p(\lambda) \, (\Pi_{\eps, 0}   +\Pi_{\eps,2+} ) \, \le\, C \sqrt{\log\widetilde \mu}\, \big(\Pi_{\eps,0}  +\Pi_{\eps,2+} \big) .
\end{align}
$\bullet$ Choosing $f(\mu,E_B) = \mu / (\varepsilon \log\widetilde \mu)$ and $g(\mu,E_B) = \varepsilon \mu$ leads to
\begin{align}
 \Pi_{\eps, 1}    \, X(\lambda) \, (\Pi_{\eps, 0}  +\Pi_{\eps, 2+}  ) \,  + \, \text{h.c.} \, \le \, C\eps ^{-1}\ \Pi_{\eps,1}  \, + \, C\eps \log\widetilde \mu \ (\Pi_{\eps,0}  + \Pi_{\eps, 2+}  ).
\end{align}
$\bullet$ For the term with $\Pi_{\eps,1} $ on both sides, the estimate in \eqref{eq: Q X Q tilde line 2} is not good enough (for obtaining an error of order one w.r.t.\ $\widetilde \mu$). A possible improvement, however, is readily obtained from 
\begin{align}
\Pi_{\eps,1} \, X_p(\lambda) \, \Pi_{\eps,1}  \, = \, \Pi_{\eps,1} \, \left(A_p(\lambda)a^*(\eta^>_\eps) + \text{h.c.} \right) \, \Pi_{\eps,1} 
\end{align}
which is true since $A_p(\lambda) (a^*(\eta) - a^*(\eta^>_\eps)) \Pi_{\eps,1}  w \in \text{Ran}(\Pi_{\eps, 0 }  )$ and $\Pi_{\eps,1}\Pi_{\eps,0} = 0$. Following now the same steps that led to \eqref{eq: Q X Q tilde line 2} and using in addition \eqref{eq: estimate eta geq}, we obtain
\begin{align}
\Pi_{\eps,1}  \, X_p (\lambda) \, \Pi_{\eps,1}   \, & \le\, \Pi_{\eps,1}  \left(  f(\mu,E_B) \, A_p(\lambda) A_p^*(\lambda)  \, + \, \frac{a(\eta^>_\eps) a^*(\eta^>_\eps )  }{g(\mu,E_B)} \right) \Pi_{\eps,1} \, \notag \\[1mm]
& \le \, C\left( \frac{f(\mu,E_B)\log \widetilde \mu}{\mu} \, + \, \frac{\mu}{g(\mu,E_B)\eps \log\widetilde \mu} \right) \, \Pi_{\eps,1} .
\end{align}
With $f(\mu,E_B) = \mu /(\sqrt \eps \log\widetilde \mu)$ and $g(\mu,E_B) =  \sqrt \eps \mu  / \log\widetilde \mu $, this provides $ \Pi_{\eps,1}  X_p(\lambda) \, \Pi_{\eps,1}   \le C \eps^{-1/2} \, \Pi_{\eps,1} $. 

Bringing the above bounds together proves \eqref{eq: first bound for localizing pol}.
\end{proof}

The goal of the next lemma is to analyze the condition $pol(\lambda) - r \ge 0$ for a given number $r\ge 0$. To see for which $\lambda$ such a bound may hold, we use the fact that this operator is given by an expression of the form $K-V^*V$ with $K= G - r $ and $V= (T-\lambda)^{-1/2}a^*(\eta)$. If $K$ is self-adjoint and $K\ge c$ for some number $c>0$, it follows easily that
\begin{align}
K - V^*V \, & = \, (K-V^*V)K^{-1}(K-V^*V) + V^* ( 1- V K^{-1} V^* ) V \nonumber \\[2mm]
& \ge \, V^* ( 1- V K^{-1} V^* ) V.\label{eq: second Birman Schwinger argument}
\end{align}
This is a key argument in the proof of the following proposition.
\begin{proposition}\label{lem: perturbed polaron equation} For any fixed $r > 0$ there exists a constant $c_0 >  0$ such that for all $L^2\vert E_B\vert \ge 1$ and $\mu /\vert E_B\vert \ge c_0$ the following implication holds: $pol(\lambda) \ge r $ if $\lambda$ satisfies
\begin{align}\label{eq: perturbed polaron equation}
E_0(\mu) - \lambda - \frac{1}{L^2} \sum_{k^2\le \mu} \frac{1}{G(0,E_0(\mu)-\lambda-k^2) - r  } =   0.
\end{align} 
\end{proposition}
\noindent \textbf{Remark 4.1.} We call \eqref{eq: perturbed polaron equation} the perturbed polaron equation.
%\begin{remark} We call \eqref{eq: perturbed polaron equation} the perturbed polaron equation since for $r = 0$ it coincides with the original polaron equation \ref{eq: polaron equation} except for the fact that we evaluate $G(p,\tau)$ at $p=0$ (the latter as we will show in Lemma \ref{} only contributes a subleading error). In Lemma \ref{} we show that for any fixed $ r \ge   0$, the perturbed polaron equation has a unique solution $\lambda(\mu,E_B)$ whose leading order is given by $E_0(\mu)+e_P(\mu,E_B)$ as  $\widetilde \mu \to \infty$ (the number $r\ge 0$ only enters the subleading contributions).
%\end{remark}
\begin{proof}
Since $T\ge E_0(\mu)-\mu$ on $\HH_{N(\mu)-1}$ and $\lambda\le E_0(\mu)+e_{\rm P}(\mu,E_B)$, it follows by \eqref{eq: operator bounds for G 2} that $G(0,T-\lambda)$ exceeds the value of $r$ for $\widetilde \mu$ large enough. For such $\widetilde \mu$ we can use \eqref{eq: second Birman Schwinger argument} to find
\begin{align}\label{eq: lower bound for pol operator}
 pol(\lambda) - r  \, \ge \, a(\eta)  \frac{1}{T-\lambda} \, \mathcal F(T-\lambda,r) \, \frac{1}{T-\lambda} a^*(\eta) \rst  \HH_{N(\mu)-1} 
\end{align}
where
\begin{align}\label{eq: definition of mathcal F}
\mathcal F(T-\lambda, r ) =  \left( T-\lambda - a^*(\eta) \frac{1}{G(0,T-\lambda) -  r } a(\eta) \right) \rst \HH_{N(\mu)}.
\end{align}
From here it follows similarly as in the proof of \cite[Lemma 4.2]{LM_InfiniteMass} that $\mathcal F(T-\lambda, r )\ge 0$ if $\lambda $ satisfies the inequality
\begin{align}
E_0(\mu) - \lambda - \frac{1}{L^2} \sum_{k^2 } \frac{1}{G(0, E_0(\mu) - k^2 - \lambda ) - r } \ge  0.
\end{align}
For convenience of the reader we provide the proof of the last statement in Appendix \ref{App: completing the proof of perturbed polaron equation}.
\end{proof}

Next we prove the existence of a unique solution to the perturbed polaron equation \eqref{eq: perturbed polaron equation} in the interval $(-\infty, E_0(\mu) + e_{\rm P}(\mu,E_B) ]$ and provide a suitable estimate for the difference of this solution and $E_0(\mu) + e_{\rm P}(\mu,E_B)$.

\begin{lemma} \label{lem: existence and bound of the perturbed polaron eqn} For any fixed $r > 0$ there exists a constant $c_0>0$ such that for all $L^2\vert E_B\vert \ge 1$ and $\mu / \vert E_B\vert \ge c_0$, the perturbed polaron equation \eqref{eq: perturbed polaron equation} admits a unique solution in the interval $(-\infty, E_0(\mu) + e_{\rm P}(\mu,E_B) ]$. We call this solution $\lambda(\mu,E_B) $.\footnote{The omission of the $r$ dependence of $\lambda(\mu,E_B)$ is justified by \eqref{eq: bound for Lambda minus lambda_1}.}  Moreover there exists a constant $C>0$ such that
\begin{align}\label{eq: bound for Lambda minus lambda_1}
E_0(\mu) +e_{\rm P}(\mu,E_B) - \lambda(\mu,E_B) \, \le \, C(1+r)   \, \frac{ \vert e_P(\mu,E_B)\vert   }{ \log (\mu / \vert E_B \vert)}
\end{align}
for all $L^2\vert E_B\vert \ge 1$ and $\mu / \vert E_B \vert \ge c_0$.
\end{lemma}

\begin{proof} To prove the existence of a solution we write \eqref{eq: perturbed polaron equation} as $E_0(\mu)-\lambda = f  (\lambda)$ with
\begin{align}
 f  (\lambda) = \frac{1}{L^2}\sum_{k^2\le \mu} \frac{1}{G( 0 , E_0(\mu)-\lambda-k^2) -  r } 
\end{align}
a continuous monotonically increasing function $f (\lambda) : (-\infty , E_0(\mu)+e_{\rm P}(\mu,E_B)]\to \mathbb R$. By definition of $G(q,\tau)$ we have
$ f (\lambda) \to 0 $ as $\lambda \to -\infty$. Next consider $f(\Lambda(\mu,E_B))$ with $\Lambda(\mu,E_B) = E_0(\mu)+e_{\rm P}(\mu,E_B)$. With the help of \eqref{eq: polaron equation},
\begin{align}
 f(\Lambda(\mu,E_B)) & = - e_{\rm P}(\mu,E_B) \notag \\[1mm]
  & \quad +  \frac{1}{L^2}\sum_{k^2\le \mu } \frac{ r + G(k , -e_{\rm P}(\mu,E_B)-k^2) - G(0, -e_{\rm P}(\mu,E_B)  -k^2)    }{( G( 0 ,-e_{\rm P}(\mu,E_B) - k^2 ) - r )\,  G(k, - e_{\rm P}(\mu,E_B) - k^2 )},\label{eq: prove of uniqueness}
\end{align}
and by way of Lemma \ref{lem: asymptotic form of G}, $G(k , -e_{\rm P}(\mu,E_B)-k^2) - G(0, -e_{\rm P}(\mu,E_B)  -k^2) \ge -2C$, we see that the second line in \eqref{eq: prove of uniqueness} is bounded from below by a constant times $-\mu /(\log \widetilde \mu)^2$. Hence for all $\widetilde \mu$ large enough, we infer $f(\Lambda(\mu,E_B)) \ge \frac{1}{2} \vert e_{\rm P}(\mu,E_B)\vert >0  $. These observations imply that there is a unique $\lambda (\mu,E_B ) \in (-\infty, \Lambda(\mu,E_B))$ such that $f(\lambda(\mu,E_B )) = E_0(\mu)-\lambda(\mu,E_B)$.

The difference between $\Lambda(\mu,E_B) $ and $\lambda(\mu,E_B)$ is estimated by
\begin{align}
& \Lambda (\mu,E_B) - \lambda(\mu,E_B)  \notag\\[2mm]
 & = \frac{1}{L^2}\sum_{k^2\le \mu}  \bigg( \frac{ r  +  G(k , - e_{\rm P}(\mu,E_B) -k^2 )  - G(0,E_0(\mu)-k^2-\lambda(\mu,E_B ) )  }  { G(k, - e_{\rm P}(\mu,E_B) - k^2 ) ( G(0,E_0(\mu)-k^2-\lambda(\mu,E_B ) ) - r) } \bigg) \notag\\
 & \le \frac{C (1+r)  }{(\log \widetilde \mu)^2} \bigg( \mu  + \frac{1}{L^2}\sum_{k^2\le \mu}  \big(  G(k^2, - e_{\rm P}(\mu,E_B) -k^2 )  - G(0,E_0(\mu)-k^2-\lambda(\mu,E_B) )   \big) \bigg),
\end{align}
where we used Lemma \ref{lem: asymptotic form of G} to estimate the denominator from below by a constant times $(\log \widetilde \mu)^2$. In the remainder we show 
\begin{align}\label{eq: bound for difference G(k) and G(0)}
\frac{1}{L^2} \sum_{k^2\le \mu} \big( G(k , -e_{\rm P}(\mu,E_B) -  k^2 )  - G(0, E_0(\mu)  - k^2 - \lambda(\mu,E_B) )   \big)  \le C\mu
\end{align}
which proves that the left side of \eqref{eq: bound for Lambda minus lambda_1} is bounded from above by $C(1+r) \mu/(\log (\mu/\vert E_B\vert))^2$. 

To verify \eqref{eq: bound for difference G(k) and G(0)} we use $\lambda(\mu,E_B) \le E_0(\mu) + e_{\rm P}(\mu,E_B)$ and again Lemma \ref{lem: asymptotic form of G} to estimate the expression inside the brackets from above by
\begin{align}
& \frac{1}{4\pi m } \log \Bigg( \frac{ \frac{k^2}{M+1} - e_{\rm P}(\mu,E_B)+  m \mu  - k^2    }{  -e_{\rm P}(\mu,E_B)+ m \mu - k^2  } \Bigg)  + 2 C.
\end{align}
With $0\le k^2 \le \mu$, $0 \le -e_{\rm P}(\mu,E_B) \le \mu$ and $\vert e_{\rm P}(\mu,E_B) \vert = O(\mu /\log\widetilde \mu)$, one further verifies that the logarithm is bounded from above by $\log  ( ( 2M^2+ 4M + 1 )/(M+1)) \le  C$.
\end{proof}

Let us summarize the result of this section.

\begin{corollary} \label{cor: analysis of Phi} For any fixed $\eps > 0$ there exist constants $c_0,C>0$ such that $\lambda (\mu,E_B) \le E_0(\mu) + e_{\rm P}(\mu,E_B)$, the unique solution of the perturbed polaron equation \eqref{eq: perturbed polaron equation} with $r= \eps^{-3}$, satisfies the following two properties: $\Phi_p(\lambda (\mu,E_B)) \ge 0$ and
\begin{align}\label{eq: bound for Lambda minus lambda}
\lambda(\mu,E_B) - E_0(\mu) - e_{\rm P}(\mu,E_B) \, \ge \, - C  (1+\eps^{-3}) \, \frac{ \vert e_P(\mu,E_B)\vert   }{ \log (\mu/\vert E_B\vert )}
\end{align}
for all $p\in \kappa \mathbb Z^2$, $  L^2 \vert E_B\vert \ge 1$ and $\mu / \vert E_B\vert  \ge c_0$.  
\end{corollary}

In the next section we show that $\Psi_p(\lambda,\varepsilon) \ge 0$ for all $\lambda \le E_0(\mu)+e_{\rm P}(\mu,E_B)$ provided that $M>1.225$ and $\varepsilon$ is sufficiently small.

%---------------------------------------------------------------------------------------------------------------------------------------------------------------
%---------------------------------------------------------------------------------------------------------------------------------------------------------------
\section{Analysis of $\Psi_p(\lambda,\varepsilon)$: stability condition \label{sec: analysis of Psi}}
%---------------------------------------------------------------------------------------------------------------------------------------------------------------
%---------------------------------------------------------------------------------------------------------------------------------------------------------------

On the subspace $\text{Ran}(\Pi^\perp_\eps)$ it is not clear how to obtain a suitable $L$-independent bound for the operator $P_p(\lambda)$. A possible solution to this difficulty is to estimate its negative part in terms of $G(p-P_{\rm f},T-\lambda)$. Such a bound was derived in \cite{Linden,GL_Stability} in the context of the 2D Fermi polaron at zero density (there the model is defined on $\mathbb R^2$ instead of the box $\Omega$ and the kinetic energy $E_0(\mu=0)$ is zero). The strategy of our proof follows the one developed there, but several new obstacles need to be dealt with in the present case. The new obstacles are due to $\mu>0$ and the fact that we have to work with momentum sums instead of integrals.

We write $P_p (\lambda) = P_p (\lambda ) - \widetilde P_p ( \lambda , \eps  ) +  \widetilde P _p (\lambda, \eps ) $ where
\begin{align}
\widetilde P_p(\lambda ,\eps  ) = \frac{1}{L^2}\sum_{k^2,l^2 > \mu / \varepsilon  } a_l^* \frac{1}{\frac{1}{M}( p -P_{\rm f} -k-l)^2+T+k^2+l^2-\lambda}a_k.
\end{align}
The operator $ P_p (\lambda)\,   - \, \widetilde P_p (\lambda,\eps ) $ is the easy part and can be estimated by the following lemma.
\begin{lemma} \label{lem: estimate for P minus P tilde}
There are constants $c_0,\varepsilon_0,C>0$ such that
\begin{align}
\label{eq: estimate for P minus P tilde}    P_p (\lambda)\,   - \, \widetilde P_p (\lambda,\eps  ) \,  \ge - C \, \sqrt{ \eps^{-1}\log( \mu / \vert E_B\vert ) } 
\end{align}
on $\HH_{N(\mu)-1}$ for all $p \in \kappa \mathbb Z^2$, $L^2\vert E_B\vert \ge 1$, $ \mu / \vert E_B\vert \ge c_0$ and $\varepsilon \in (0,\eps_0)$.
\end{lemma}
\begin{proof} Write
\begin{align}
P_p (\lambda)\,   - \, \widetilde P_p(\lambda,\eps) & = \frac{1}{L^2}\sum_{ \mu <  k^2,l^2 \le  \mu / \eps } a_l^* \frac{1}{\frac{1}{M}( p -P_{\rm f}-k-l)^2+T+k^2+l^2-\lambda}a_k \label{eq: first line of P minus tilde P}\\
& \quad + \frac{1}{L^2}\sum_{ \substack{ \mu <  k^2 \le   \mu / \eps \\ l^2 > \mu / \eps} } a_l^* \frac{1}{\frac{1}{M}( p -P_{\rm f}-k-l)^2+T+k^2+l^2-\lambda}a_k + \text{h.c}. \label{eq: second line of P minus tilde P}
\end{align}
We proceed as in the proof of \eqref{bound: Pi P Pi} to obtain
\begin{align}
\eqref{eq: first line of P minus tilde P} 
%& \ge - \frac{f(\mu,E_B)}{2}  \Bigg( \frac{1}{L^2}\sum_{\mu \le k^2,l^2 <   \mu / \eps } a_l a_k^* \Bigg)  - \frac{1}{2f(\mu,E_B)   }  \Bigg( \frac{1}%{L^2} \sum_{\mu \le k^2,l^2 <   \mu / \varepsilon }  a_k \frac{1}{( T-\lambda)^2}  a_l^* \Bigg) \nonumber \\
& \ge -C \bigg( f(\mu,E_B) \frac{\mu}{\eps} + \frac{1}{ f(\mu,E_B)  \vert e_{\rm P}(\mu,E_B)\vert } \bigg).
\end{align}
%where we made use of \eqref{eq: general sum for k^2 between a mu and b mu} and estimated the second summand as in \eqref{eq: bound for Pi P Pi line 3}. 
Choosing $f(\mu,E_B)   = \sqrt {\varepsilon^{-1}  \log \widetilde \mu} / \mu $ we get the desired bound for this line. The proof for the second line works in complete analogy.
%\begin{align}
%\eqref{eq: second line of P minus tilde P} & \ge f(\mu,E_B)  \Bigg( \frac{1}{L^2}\sum_{\mu \le k^2,l^2 < \alpha \mu } a_l a_k^* \Bigg)  - \frac{1}{f(\mu,E_B)}  \Bigg( \frac{1}{L^2} \sum_{\alpha \mu  < k^2,l^2 }  a_k \frac{1}{( \frac{1}{M} P_f^2 + T-\lambda)^2}  a_l^* \Bigg) \nonumber \\
%& \ge -C \bigg( \alpha f(\mu,E_B) \mu + \frac{1}{f(\mu,E_B) \vert e_P(\mu,E_B)\vert } \bigg) .
%\end{align}
\end{proof}

To a large extent this section is about the derivation of a lower bound for $\widetilde P_p(\lambda,\eps)$. For that purpose we need to introduce some further notation. For $\eps$ small enough, let the function $\beta(\cdot, \eps) :[0,1]\to (0,1]$ be given by
\begin{align}\label{def: definition of beta(u)}
 \beta(u,\eps  ) = \min\Bigg\{ 1, \frac{(M+1-u) (M+2)   \big( 1-  (1+ \frac{ M+1-u }{ M (M+2)  } ) \sqrt \eps  \big) }{ (M+1-u)(1- 2 \sqrt \eps  ) + M(M+2)(1-\sqrt \eps )} \Bigg\}
\end{align}
%\begin{align}\label{def: definition of beta(u)}
%\beta(u) := \min\left\{ 1, \frac{(M+1-u)(M+2)}{M^2+3M+1-u}\right\},\qquad u\in [-1,1],
%\end{align}
and set
\begin{align}\label{eq: definition of C(m,alpha)}
\alpha (M,\eps) = \frac{1}{2 }\left(\frac{1}{M(1- \sqrt \eps )+1} + \int_0^1 \frac{1}{\beta(u,\eps  ) (M(1- \sqrt \eps ) +1-u)} \D u\right).
\end{align}

\begin{proposition}\label{lem: stability at positive density} There are constants $c_0,\varepsilon_0 ,C>0$ such that
\begin{align} 
\widetilde P_p(\lambda,\varepsilon )\,  \ge\,  - \frac{1}{ 1 - \varepsilon }\, \bigg( \frac{\alpha (M,\varepsilon  )}{4\pi }\,  \log\left( 1 + \frac{T-  \lambda +2\mu }{  \mu}\right)   +  \frac{C}{ \sqrt{\mu/\vert E_B \vert }}   \bigg)  
\end{align}
on $\HH_{N(\mu)-1}$ for all $p\in \kappa \mathbb Z^2$, $L^2\vert E_B\vert \ge 1$, $ \mu/ \vert E_B \vert \ge c_0$ and $\varepsilon \in (0,\varepsilon_0)$.
\end{proposition}

To prove Proposition \ref{lem: stability at positive density} we combine the next two lemmas.

\begin{lemma}\label{lem: T_> vs T bound} Let $T_{> \mu / \eps  } = \sum_{k^2> \mu / \eps } k^2 \, a_k^* a_k$.
For any $\eps \in (0,1) $, it holds that 
\begin{align}\label{eq: T_> vs T bound}
  \frac{T_{> \mu/\eps }}{T-\lambda+2\mu}  \restriction \HH_{N(\mu)-1}  \,  \le \, \frac{1}{1- \eps  }.
\end{align}
\end{lemma}
\begin{lemma}\label{lem: stability at positive density involving Q} There are constant $c_0 , \eps_0, C>0$ such that
\begin{align}\label{eq: estimate stability at positive density involving Q}
 \widetilde P_p (\lambda,\eps ) \ge -   \frac{T_{>  \mu / \eps  }}{T-\lambda+2\mu }\,  \bigg( \frac{ \alpha (M, \eps  ) }{4\pi } \,  \log\left( 1 + \frac{T-  \lambda + 2\mu }{  \mu}\right)   +  \, \frac{C}{ \sqrt{\mu/\vert E_B \vert }}  \bigg)  
\end{align}
on $\HH_{N(\mu)-1}$ for all $p\in \kappa \mathbb Z^2$, $L^2\vert E_B\vert \ge 1$, $\mu/ \vert E_B\vert \ge c_0$ and $\eps\in (0,\eps_0)$.
\end{lemma}
\begin{proof}[Proof of Proposition \ref{lem: stability at positive density}] Since $T-\lambda \restriction \HH_{N(\mu)-1} \ge  - \mu $, the operator 
\begin{align}
  \log\left( 1 + \frac{T-  \lambda + 2\mu }{  \mu}\right)    \restriction \HH_{N(\mu)-1} \ge \log(2)
\end{align}
is positive. Since $T$ and $T_{>\mu/\eps}$ commute, we can use Lemma \ref{lem: T_> vs T bound} to prove Proposition \ref{lem: stability at positive density} with the aid of \eqref{eq: estimate stability at positive density involving Q}.
\end{proof}

\begin{proof}[Proof of Lemma \ref{lem: T_> vs T bound}] Using again $T-\lambda \restriction \HH_{N(\mu)-1} \ge  - \mu $
in combination with $0\le T_{>  \mu / \eps } \le T $, the operator
\begin{align}
 \frac{ T_{> \mu / \eps  } }{T-\lambda+2\mu }   \le  1 + \frac{\lambda-2\mu  }{T-\lambda+2\mu } \le 1 + \frac{E_0(\mu)-2\mu}{\mu}
\end{align}
is bounded when restricted to $\HH_{N(\mu)-1}$. Hence it is sufficient to show \eqref{eq: T_> vs T bound} on the dense subspace $\text{lin} D \subseteq \HH_{N(\mu)-1}$ given by all finite linear combinations of anti-symmetric products of plane waves, see \eqref{eq: definition of total set D}. Since the states in $ \text{lin}  D$ are linear combinations of simultaneous eigenstates of $T\rst\HH_{N(\mu)-1}$ and $T_{> \mu / \eps } \rst\HH_{N(\mu)-1}$, we can restrict the argument further to the set $D$ itself. This becomes particularly useful when writing $ D = \bigcup_{n\ge 0} W_n(\eps ) $ with   
\begin{align}\label{def: W_n sets}
 W_n(\eps) = \Big \{  w \in   D \, :   \,  \Big(\sum_{ k^2> \mu / \eps } a_k^*a_k \Big) w  = n w \Big\} 
 \end{align}
the set of anti-symmetric products of plane waves with exactly $n$ momentum modes occupied in $\{k \in \kappa \mathbb Z^2 : k^2 > \mu / \eps \}$.

Since $T_{>\mu/\eps} w = 0 $ for $w\in W_0(\eps)$, we consider $w\in W_n(\eps)$, $n\ge 1$, with $\snorm{w}^2=1$. We call $\beta( w ) $ the eigenvalue of $T$ and $\gamma ( w) $ the eigenvalue of $T_{>  \mu / \eps}$. It follows that
\begin{align}\label{eq: lower bound for T_> wrt w_n}
\gamma ( w  )  =  \sprod{w }{T_{>\mu  / \eps  } w }  >  n \,  \mu \eps ^{-1}
\end{align}
as well as
\begin{align}\label{eq: lower bound for T_<}
\beta ( w  ) - \gamma ( w  )  =  \sprod{ w }{( T - T_{>  \mu / \eps } ) w  }  \ge E_0(\mu) - (n+1) \, \mu.
\end{align}
To derive the last inequality we denote the eigenvalues of $-\Delta$ by $\lambda_i(-\Delta)$ ($i\ge 1$, numbered with increasing order and counting multiplicities) and use
\begin{align}
 \sprod{w  }{( T - T_{> \mu / \eps  } ) w }   =  \sprod{w  }{ ( \sum_{k^2\le \mu / \eps } k^2 a_k^* a_k ) w  }  
 %\ge \sum_{i=1}^{N(\mu)-1-n} \lambda_i(-\Delta) \notag \\
 &  \ge  E_0(\mu) - \sum_{i=N(\mu)-n}^{N(\mu)} \lambda_i(-\Delta).
\end{align}
From $\lambda_i(-\Delta) \le \mu$ for $i\le N(\mu)$, we obtain \eqref{eq: lower bound for T_<}. The latter together with $\lambda \le E_0(\mu)$ implies
\begin{align}
\beta (  w  ) -\lambda \, \ge  \,    \gamma ( w  )  - (n+1) \, \mu,
\end{align}
and combining this with \eqref{eq: lower bound for T_> wrt w_n}, we get
\begin{align}
\sprod{w }{\frac{T_{> \mu / \eps }}{T-\lambda+2\mu } w } = \frac{\gamma ( w ) }{\beta ( w ) -\lambda+2\mu } \le \frac{ \gamma ( w )   }{\gamma ( w )   - (n- 1) \, \mu} \le \frac{n\,  \mu \eps ^{-1} }{n\, \mu \eps^{-1} - (n-1)\, \mu} .
\end{align}
Since the expression on the right does not exceed the value $\frac{1}{1 -\varepsilon }$, we have proven the statement.
\end{proof} 
 
\begin{proof}[Proof of Proposition \ref{lem: stability at positive density}] We start by introducing the abbreviations
\begin{align}
\widehat k  = k + \frac{1}{M+2} (p - P_{\rm f}) , \quad \quad  \widehat l  = l + \frac{1}{M+2} (p -P_{\rm f})
\end{align}
by which one writes the denominator in the expression defining $\widetilde P (\lambda,\eps)$ as
\begin{align}
m  (\widehat k ^2+ \widehat l^2) + \frac{2}{M} \widehat k \cdot \widehat l + \frac{1}{M+2}( p  - P_{\rm f})^2 +  T - \lambda.
\end{align}
For $w \in \HH_{N(\mu)-1}$ we define $\widetilde w  \in L^2( \kappa  \mathbb Z^2;\HH_{N(\mu)-2})$ by $\widetilde w(k) =a_kw$. Moreover we define the unitary operator $U \in \LL (L^2( \kappa  \mathbb Z^2;\HH_{N(\mu)-2})) $ by\footnote{Note that we omit the $p$-dependence of the unitary operator $U$.}
\begin{align}
(U\varphi)(k;k_1,...,k_{N(\mu)-2} ) = \varphi \big( k+\frac{1}{M+2}\big( p - \sum_{i=1}^{N(\mu) -2} k_i \big); k_1,...,k_{N(\mu)-2} \big),
\end{align} 
where we use the notation $(U\varphi)(k ;k_1,...,k_{N(\mu)-2})$ for the Fourier space representation of $(U\varphi)(k) \in \HH_{N(\mu)-2}$. With these definitions at hand, it is not difficult to compute
\begin{align}\label{eq: P tilde in terms of sigma line 0}
\sprod{w}{\widetilde P_p(\lambda, \eps)w} 
%&  =\, \frac{1}{L^2}\sum_{k,l } \sprod{ (\chi_{(\mu / \eps,\infty)}  \widetilde w)(k) }{  \big( m  (\widehat k^2+\widehat l^2) + \frac{2}{M} \widehat k %\cdot \widehat l + \frac{ ( p - P_{\rm f})^2}{M+2} +T  - \lambda \big)^{-1}  ( \chi_{(\mu / \eps,\infty)}  \widetilde w)(l) } \notag \\[0.5mm]
&  =\, \frac{1}{L^2}\sum_{k,l } \sprod{ (\chi_{ \mu / \eps } \widetilde w)(k) }{ U \sigma(k,l) U^* ( \chi_{\mu / \eps} \widetilde w)(l ) } 
\end{align}
where $\chi_{(\mu / \eps,\infty)}$ stands for the characteristic function $k\mapsto \chi_{(\mu / \eps ,\infty)}(k^2)$ and where 
\begin{align}
\sigma(k,l) =  \frac{1}{L^2}\, \frac{1}{  m ( k^2+ l^2) + \frac{2}{M}  k \cdot  l + \frac{1}{M+2} (p -P_{\rm f})^2 +  T - \lambda}.
\end{align}
Denoting the scalar product on $L^2( \kappa \mathbb Z^2 ;\HH_{N(\mu)-2})$ by $\langle  \hspace{-0.75mm}   \langle \cdot, \cdot \rangle  \hspace{-0.75mm}  \rangle$, \eqref{eq: P tilde in terms of sigma line 0} is rewritten as
\begin{align}\label{eq: P tilde in terms of sigma}
\sprod{w}{\widetilde P_p  (\lambda, \eps )w} \, = \, \langle \hspace{-0.75mm} \langle \chi_{(\mu / \eps ,\infty)} \widetilde w , U\, \sigma\, U^* \chi_{(\mu / \eps ,\infty)} \widetilde w \rangle \hspace{-0.75mm} \rangle
\end{align}
where $\sigma$ is the operator on $L^2( \kappa  \mathbb Z^2;\HH_{N(\mu)-2})$ with operator-valued kernel $\sigma(k,l)$. Next, we show that the negative part of $\sigma$ has the kernel $\sigma^-(k,l) = \frac{1}{2} \left( \sigma(-k,l) - \sigma(k,l)\right)$. To this end, consider the reflection operator $R$ defined by $(R \widetilde w)(k) = \widetilde  w(-k)$ for any $\widetilde w\in L^2( \kappa  \mathbb Z^2;\HH_{N(\mu)-2})$. It is straightforward to verify $R \sigma = \sigma R$. Moreover, $R \sigma$ is a positive operator, which can be seen as follows. The integral kernel of $R\sigma$ is given by $(R\sigma)(k,l) = \sigma(-k,l)$ and has the integral representation
\begin{align}
\sigma( - k, l) & \, = \, \frac{1}{L^2} \, \int_0^\infty \,  e^{-tk^2} \left( e^{-t(k-l)^2/M} \ e^{-t(\frac{1}{M+2} (p -  P_{\rm f})^2 +  T - \lambda)} \right) e^{-t l^2} \, \D t.
\end{align}
Then use the following identity for $\psi \in L^2(\Omega)$ and its Fourier transform $\widehat \psi \in \ell^2(\kappa \mathbb Z^2)$, 
\begin{align}
\frac{1}{L^2} \sum_{k,l} \overline{\widehat \psi(k)} \ e^{-t(k-l)^2/M}\ \widehat \psi(l) \ = \ \int_\Omega  \vert \psi(x)\vert^2 \sum_k e^{ikx} e^{-tk^2/M} \, \D x.
\end{align}
This together with Poisson's summation formula (see e.g.\ \cite[Section 3.2]{Grafakos2014}) and the fact that the Fourier transform of a Gaussian is a positive function implies that $R\sigma$ is a positive operator. Consequently, we have $R\sigma = \vert \sigma\vert$ since $R\sigma$ is positive and $ \sigma^2 = (R\sigma)(R\sigma)$. The positive and negative parts of $\sigma$ are thus given by $\sigma^{\pm} = \pm (\sigma \pm R \sigma)/2$ and the corresponding kernels by $\sigma^{\pm}(k,l) = \pm (\sigma(k,l) \pm  \sigma( - k ,  l))/2 $. 

We proceed by writing the kernel of the negative part as $\sigma^-(k,l) =  \frac{1}{2}\int _{-1}^1 \frac{d}{du} \sigma(- u k, l ) \D u$, and hence
\begin{align}
\sigma^-(k,l) 
% = \frac{1}{2L^2} \, \frac{1}{ m  ( k^2+ l^2) - \frac{2u}{M}  k \cdot  l + \frac{1}{M+2} (P_f-P) ^2 +  T - \lambda} \bigg\vert_{u=-1}^{u=1} \notag
% \\[1mm]
%& = \frac{1}{2L^2} \int\limits_{-1}^{1} du \frac{d}{du}  \frac{1}{ m  ( k^2+ l^2) - \frac{2u}{M}  k \cdot  l + \frac{1}{M+2} (P-P_f)^2 +  T - \lambda} %\notag  \\[1mm]
& = \frac{M  k \cdot l}{L^2} \int\limits_{-1}^{1}   \frac{1}{[ (M+1) (k^2+l^2) - 2 u k\cdot l + B ]^2} \D u  ,
\end{align}
where $B = \frac{M}{M+2}( p -  P_{\rm f})^2 + M (T -\lambda)$. Using this in combination with \eqref{eq: P tilde in terms of sigma}, we get
\begin{align}
\widetilde P_p (\lambda, \eps ) \, \ge \, - \frac{M}{L^2} \sum_{k^2,l^2>  \mu / \eps } a_k^* \left(\ \int\limits_{-1}^{1}  \frac{\widehat k \cdot \widehat l}{\big[ (M+1) ({\widehat k}^2+ {\widehat l}^2) - 2 u \widehat k\cdot \widehat l + B \big]^2} \D u \right) a_l.
\end{align}
To the expression on the right we apply the following inequality which is a version of the Schur test and is easily proven by applying the Cauchy-Schwarz inequality two times,
\begin{align}
 \sum_{k^2,l^2>  \mu / \eps } a_k^*\, J(k,l)\, a_l \, \le \, \sum_{k^2 >  \mu / \eps } k^2\, a_k^* \left( \sum_{ l^2>  \mu / \eps  } \frac{\vert J (k,l)\vert}{l^2} \right) a_k
\end{align}
for any family of bounded operators $(J(k,l))_{k,l\in \kappa \mathbb Z^2}$ on $\mathcal F$ satisfying $J(k,l)^* = J(k,l)$. This provides
\begin{align}
\widetilde P_p(\lambda, \eps ) \, \ge \,  - M \sum_{k^2>  \mu / \eps } k^2\, a_k^* \underbrace{\left(   \frac{1}{L^2} \sum_{l^2>  \mu / \eps} \ \int\limits_{-1}^{1}   \frac{\vert \widehat k \cdot \widehat l \vert }{l^2  \big[ (M+1) (\widehat k^2+ \widehat l^2) - 2 u \widehat k\cdot \widehat l + B \big]^2}\D u   \right) }_{= f(k^2,p-P_{\rm f},T)   }  a_k \label{eq: definition of f}
\end{align}
as an operator inequality on $\HH_{N(\mu)-1}$. Our next goal is to find a suitable function $g$ such that for $k^2> \mu / \varepsilon $, we have $f(k^2,p-P_{\rm f},T) \le g(T+k^2)$ on $\HH_{N(\mu)-2}$. For such a function we have
\begin{align}
\widetilde P_p (\lambda , \eps )  & \, \ge\,  -   M   \sum_{k^2>  \mu / \eps } k^2\,  a^*_k \, g(T+k^2)\, a_k\,  
%& \ge -  \, M \Bigg( \sum_{k^2>  \mu / \eps } k^2 \,  a^*_k a_k \Bigg) g(T)    \, 
\ge \,  -  M  T_{> \mu / \eps } \ g(T) \, \label{eq: bound for P with g}
\end{align}
since $g(T+k^2) a_k = a_k g(T)$ and $T_{> \mu / \eps } = \sum_{k^2>  \mu / \eps }  k^2 a_k^* a_k $.

To find a suitable function $g$, it is helpful to check that the expression inside the square brackets in the denominator in \eqref{eq: definition of f} is positive for $\eps $ small enough. To see this we use 
\begin{align}\label{eq: lower bound for widehat k squared}
\widehat k^2 \ge \sqrt \eps  k^2 -  \frac{ \sqrt \eps}{ 1 - \sqrt \eps }\, \frac{(p-P_{\rm f})^2}{(M+2)^2} , 
\end{align}
and similarly for $\widehat l^2$, together with $k^2,l^2 > \mu/\eps$ and $T-\lambda\ge -2\mu $ on $\HH_{N(\mu)-2}$ to find
\begin{align}
& (M+1) (\widehat  k^2+ {\widehat l}^2) - 2 u \widehat k\cdot \widehat  l + B  \notag \\[1mm]
 &\qquad  \ge 2 M (\eps^{-1/2} -1 ) \mu +  \bigg( 1 - \frac{ \sqrt \eps}{ 1 - \sqrt \eps }\, \frac{1}{ M+2 } \bigg) \frac{M }{M+2}(p- P_{\rm f})^2  \label{eq: positive denomiominator}.
\end{align}
Next we use that $ -2 u \widehat k\cdot \widehat l \ge 0$ either for $u\in [-1,0]$ or for $u\in [0,1]$. This makes the quotient in the definition of $f$ larger and also independent of $u$ on the respective interval. On the other interval, we employ $ 0 \ge  - 2 u \widehat k\cdot \widehat l \ge - \vert  u \vert (\widehat k^2 + \widehat l^2)$. In both cases this leads to
\allowdisplaybreaks
\begin{subequations}
\begin{align}
\int\limits_{-1}^{1}   \frac{\vert \widehat k \cdot \widehat l \vert }{l^2 \big[ (M+1) (\widehat k^2+ \widehat l^2) - 2 u \widehat k\cdot \widehat l + B \big]^2} \D u  & \le  \frac{\vert \widehat k \cdot \widehat l \vert}{l^2 \big[ (M+1) (\widehat k^2+ \widehat l^2) + B \big]^2} \label{eq: estimating f line one}\\
&\hspace{-1cm} + \int\limits_{0}^{1}   \frac{\vert \widehat k \cdot \widehat l \vert}{l^2 \big[ (M+1-u) (\widehat k^2+ \widehat l^2)   + B \big]^2} \D u  . \label{eq: estimating f line two}
\end{align}
\end{subequations}
In the denominators we proceed with the bound
\begin{align}\label{eq: Bound for the denominator in f}
(M+1-u) ( \widehat k^2+\widehat l^2) + B  \ge 2\, \vert \widehat k \cdot \widehat l \vert  \big(  M (1- \sqrt \eps ) +1-u  \big) .
\end{align}
The latter is verified by
\begin{align}
& (M+1-u) ( \widehat k^2+\widehat l^2) + B  \notag\\
 & \quad \ge \bigg(  \widehat k^2+\widehat l^2  + \frac{ (p- P_{\rm f})^2 M }{ (M+1)(M+2)   } \bigg) \bigg(  M+1-u  - \frac{2\mu M }{ \widehat k^2+\widehat l^2  + \frac{ (p-P_{\rm f})^2 }{(M+1)(M+2)}} \bigg)
\end{align}
on $\HH_{N(\mu)-2}$ in combination with
\begin{align}
\widehat k^2+\widehat l^2  + \frac{ (p-P_{\rm f})^2 M }{( M+1 ) (M+2 )  }> \frac{2\mu }{\sqrt \eps}
\end{align}
which, in turn, follows from \eqref{eq: lower bound for widehat k squared} and $k^2+l^2\ge 2\mu/\eps$. Putting the different steps together, one obtains
\begin{align}\label{eq: estimate f vs g line 01}
f(k^2,p-P_{\rm f},T) &  \le \widetilde f (k^2,p-P_{\rm f},T,0) + \int_0^1  \widetilde f (k^2,p-P_{\rm f} ,T,u) \D u
\end{align}
with
\begin{align}\label{eq: bound}
\widetilde f (k^2,p-P_{\rm f},T,u) = \frac{1}{L^2}\sum_{l^2> \mu / \eps }  \frac{1}{2l^2 ( M(1- \sqrt{\eps}  ) +1-u)\big[ (M+1-u) (\widehat k^2+ \widehat l^2) + B  \big]} .
\end{align} 
In the expression inside the square brackets we estimate $\widehat k^2$ and $\widehat l^2$ by \eqref{eq: lower bound for widehat k squared} to get the lower bound
\begin{align}
& \big[ ... \big] \ge (M+1-u) \bigg( \sqrt \eps l^2 +\sqrt \delta k^2 - \frac{2\mu M }{(M+1-u)}  \bigg) + M ( T - \lambda + 2\mu) \notag\\
& \hspace{1.5cm} +  \bigg( M(M+2) - \frac{\sqrt \eps}{1-\sqrt \eps } (M+1-u) - \frac{\sqrt \delta}{1-\sqrt \delta } (M+1-u)  \bigg) \frac{(p-P_{\rm f})^2}{(M+2)^2}.
\end{align}
Requiring that the second line vanishes implies
\begin{align}
\sqrt \delta = \frac{M(M+2)(1-\sqrt \eps) - \sqrt \eps (M+1-u)}{M(M+2)(1-\sqrt \eps) + (M+1-u)(1-2\sqrt \eps)} .
\end{align}
Hence we can bound the expression in square brackets by
\begin{align}
(M+1-u) (\widehat k^2+ \widehat l^2) + B  &  \ge \frac{\sqrt{\eps }}{2} l^2 (M+1-u)  + M \beta(u,\eps ) (T+k^2-\lambda+2\mu)
\end{align}
with
\begin{align}
\beta( u , \eps ) = \min\{ 1, \sqrt \delta (M+1-u)/M \}.
\end{align}
Note that for $\eps$ small enough $\beta(0,\eps  ) = 1$. Applying this to \eqref{eq: bound}, we obtain
\begin{align}
& \widetilde f (k^2,p-P_{\rm f},T,u) \notag \\
&\le \frac{1}{ 2( M (1- \sqrt \eps ) + 1 - u )  }\, \frac{1}{L^2}\sum_{l^2> \mu / \eps }  \frac{l^{-2}}{ \frac{\sqrt \eps}{2} l^2(M+1-u) + M  \beta(u,\eps) (T+k^2-\lambda+2\mu) }  .\label{eq: f tilde}
\end{align} 
Here we sum a non-negative and monotonically decreasing function of $l^2$ so that we can apply Lemma \ref{lemma: replacing sums by integrals}. To follow the next steps with more ease, let us write
\begin{align}\label{eq: f tilde 2}
\eqref{eq: f tilde} = \frac{1}{X} \, \bigg( \frac{1}{L^2} \sum_{l^2 > \mu/\eps} \frac{1}{l^2(1+Y l^2)}\bigg)
\end{align} 
with (all understood as operator on $\HH_{N(\mu)-2}$)
\begin{align}
X =  2(M(1-\sqrt \eps ) + 1 - u )  Z  , \quad Y = \frac{\sqrt \eps (M+1-u ) }{2  Z} ,
\end{align}
and $Z  = M \beta (u,\eps) (  T+k^2-\lambda + 2 \mu )$. Since for $b>0$
\begin{align}\label{eq: bound for P integral log term}
\int\limits_{\sqrt{\mu/\eps }}^\infty \frac{1}{ s (1+ b s^2) } \D s = \frac{1}{2}\log\bigg( 1 + \frac{\eps}{\mu \, b}\bigg) ,\quad 
\int\limits_{ \sqrt{\mu  / \eps } }^\infty \frac{ 1 }{ s^2 (1+ b s^2) } \D s
%= - \bigg[\frac{\sqrt b\, s\, \arctan ( \sqrt b \, s) + 1 }{s} \bigg]_{\sqrt{\mu/ \eps }}^\infty 
\le \frac{1}{ \sqrt{ \mu / \eps } } + \frac{\pi \sqrt b }{2 },
\end{align}
we obtain the bound
\begin{align}
 \frac{1}{L^2} \sum_{l^2 >  \mu/\eps} \frac{1}{l^2(1+Y l^2)} & \le \frac{1}{4\pi} \log\bigg( 1 + \frac{\eps}{\mu \, Y}\bigg) \notag \\
 & \quad +  \frac{2}{\pi L } \bigg( \frac{1}{ \sqrt{ \mu / \eps } } + \frac{\pi }{2  }  \sqrt Y \bigg) + \bigg( \frac{4\sqrt{\mu/\eps}}{\pi L} + \frac{6}{L^2}\bigg) \frac{1}{\frac{\mu}{\eps} ( 1 + \frac{\mu}{\eps} Y )} .
\end{align}
Using $T-\lambda\ge  -2\mu $ on $\HH_{N(\mu)-2}$, $k^2\ge \mu/ \eps$ and $L^2 \vert E_B\vert  \ge 1$, the second line is easily seen to be bounded by a constant times $\widetilde \mu^{-1/2}$. In the first line, we estimate
\begin{align}
\log\bigg( 1 + \frac{2\sqrt \eps M \beta (u, \eps) (T+k^2-\lambda + 2\mu )}{(M+1-u)\mu }\bigg) \le \log\bigg( 1 + \frac{T+k^2-\lambda + 2\mu }{\mu }\bigg) 
\end{align}
by choosing $\eps$ sufficiently small. This together with \eqref{eq: f tilde 2} leads to
\begin{align}
\widetilde f (k^2,p-P_{\rm f},T,u) & \le\frac{1}{T+k^2-\lambda+2\mu}\bigg( \frac{ \log \big( 1 + \frac{  T+k^2-\lambda + 2\mu }{\mu} \big) }{8\pi M ( M(1-\sqrt \eps ) +1-u) \beta(u,\eps) } + \frac{C}{\sqrt{\widetilde \mu}} \bigg)  .\label{eq: last estimate for f tilde}
\end{align}
Recalling definition \eqref{eq: definition of C(m,alpha)} for $\alpha(M,\eps)$, we set
\begin{align}
g(T) = \frac{1}{T-\lambda+2\mu} \bigg( \frac{\alpha(M,\eps)}{4\pi M}  \log \left( 1 + \frac{  T -\lambda + 2\mu }{\mu} \right)  + \frac{C}{\sqrt{\widetilde \mu}}\bigg)
\end{align}
for some suitable constant $C$. In view of \eqref{eq: estimate f vs g line 01} and \eqref{eq: last estimate for f tilde}, it follows that $f(k^2,p-P_{\rm f},T)\le g(T+k^2)$ as desired. With the aid of \eqref{eq: bound for P with g} this leads to
\begin{align}
\widetilde P_p(\lambda,\eps ) \, & \ge \, -  \frac{T_{> \mu / \eps }}{T- \lambda +2\mu }  \left( \frac{ \alpha(M,\eps  )}{4\pi } \, \log\left(1  + \frac{T-\lambda - 2\mu }{\mu} \right) +  \frac{C}{\sqrt{ \widetilde \mu} } \right)
\end{align} 
for some constant $C>0$ and thus the proof of the lemma is complete.
\end{proof}

The next statement is the main result of this section. Let us mention that the condition $M>1.225$ enters as a technical assumption and is not expected to be optimal.

\begin{corollary}\label{cor: final bound for Psi(lambda)} Let $M>1.225$ and $\lambda \le E_0(\mu)+e_{\rm{P}}(\mu,E_B)$. There exist constants $c_0,\eps_0>0$ such that $\Psi_p(\lambda,\eps)\ge 0$ for all $p\in \kappa  \mathbb Z^2$, $L^2 \vert E_B\vert \ge 1 $, $\mu/ \vert E_B\vert \ge c_0$ and $\eps \in (0,\eps_0)$.
\end{corollary} 
\begin{proof} Recalling the definition of $\Psi_p(\lambda, \eps)$ in \eqref{def: Psi of lambda}, we write 
\begin{align}
\Psi_p(\lambda,\eps)   = \Pi^\perp(\eps) ( \Psi_{p,1}(\lambda,\eps) +  \Psi_{p,2}(\lambda,\eps ) ) \Pi^\perp(\eps)
\end{align}
with
\begin{align}
\Psi_{p,1}(\lambda,\eps ) & =   \eps^{1/3}\, G(p-P_{\rm f},T-\lambda)  +  P_p(\lambda) - \widetilde P_p(\lambda,\eps )   -  K(\eps,\widetilde \mu ) - d \label{eq: Psi line two}, \\[2mm]
 \Psi_{p,2}(\lambda,\eps) & \, = \,  (1- \eps^{1/3}) G(p-P_{\rm f},T-\lambda) \, + \,\widetilde P_p(\lambda,\eps  )  + d  ,  \label{eq: Psi line one} 
\end{align}
where $K(\eps ,\widetilde \mu) =  \eps^{-1} + \eps^{-1/2} \sqrt {\log \widetilde \mu} + \eps^{1/2} \log\widetilde \mu$ and $d>0$ is a constant that we choose large enough but fixed w.r.t.\ all parameters.

By means of inequality \eqref{eq: operator bounds for G 2} and Lemma \ref{lem: estimate for P minus P tilde}, we estimate
\begin{align}
&  \Psi_{p,1}(\lambda,\eps)  \ge  \frac{\eps^{1/3} }{4\pi m } \log \widetilde \mu -  C \big(\eps^{1/3} + d  + \sqrt{ \eps^{-1} \, \log\widetilde \mu} +  K(\eps , \widetilde \mu) \big)  
\end{align}
for some ($\eps$-independent) $C>0$. With $\eps_0>0$ small enough, the right side is non-negative for all $\widetilde \mu$  large enough.

In line \eqref{eq: Psi line one} we apply \eqref{eq: operator bounds for G} to get
\begin{align}
 (1-  \varepsilon^{1/3}) G(p-P_{\rm f},T-\lambda) +   \frac{d}{2} 
% & \ge   \frac{  1-  \eps^{1/3} }{4\pi m }  \log\left( \frac{ T-\lambda +  \mu }{ \vert E_B \vert } \right)  - C + \frac{1}{2\eps} \notag \\
& \ge     \frac{ 1 - \eps^{1/3}   }{4\pi m }    \log\left( \frac{  T-\lambda + m \mu }{ \vert E_B\vert } \right) 
\end{align}
on $\HH_{N(\mu)-1}$. Since $m=1+\frac{1}{M}$, $\mu / \vert E_B\vert \ge c_0 \ge 2M$ as well as $(T-\lambda + \mu) \restriction \HH_{N(\mu)-1}\ge 0$, we can estimate the logarithm further by
\begin{align}
\log\left( \frac{  T-\lambda + m \mu }{ \vert E_B\vert } \right) \ge \log\left( \frac{ T-\lambda + \mu + c_0 \mu /M  }{ \mu } \right) \ge \log\left( \frac{ T-\lambda + 3 \mu }{ \mu } \right) .
\end{align}
Proposition \ref{lem: stability at positive density} gives a bound for the second term in \eqref{eq: Psi line one}, 
\begin{align}
  \widetilde P_p(\lambda,\eps ) +\frac{d}{2} & \, \ge \,   -  \frac{1}{1-\eps}\,  \frac{\alpha (M,\eps)}{4\pi }  \log\bigg( 1 + \frac{T-\lambda + 2\mu }{\mu}\bigg)  . \label{eq: bound in the conclusion for Psi}
\end{align}
Adding both estimates together leads to
\begin{align}
 \Psi_{p,2}(\lambda,\eps) & \, \ge \,    \bigg(  \frac{ 1 - \eps^{1/3}   }{4\pi m }   -  \frac{1}{1-\eps}\,  \frac{\alpha (M,\eps)}{4\pi } \bigg)  \log\bigg( 1 + \frac{T-\lambda + 2\mu }{\mu}\bigg) 
\end{align}
on $\HH_{N(\mu)-1}$.

The condition $\Psi_p(\lambda,\varepsilon)\ge 0$ is thus satisfied if
\begin{align}\label{eq: mass stability condition}
   (1-  \eps^{1/3} ) \, \frac{M}{M+1}  -  \frac{1}{1-\eps}\,  \alpha(M,\eps )  \, \ge \, 0.
\end{align}
This is similar to the stability condition at zero density that was derived in \cite{GL_Stability}. There it was shown that the Fermi polaron defined on $\mathbb R^2$ is stable if
\begin{align}\label{eq: stability condition at zero density}
\frac{M}{M+1} - \alpha(M,0) \ge 0
\end{align}
which was proven to hold for all $M>1.225$ \cite[Theorem 1]{GL_Stability}. Since $\alpha (M,\eps)$ depends continuously on $\eps$, we can conclude that \eqref{eq: mass stability condition} holds for any given $M>1.225$ if we choose $\eps$ sufficiently small. This completes the proof of the corollary.
\end{proof}

%---------------------------------------------------------------------------------------------------------------------------------------------------------------
\subsection{Proof of Theorem \ref{thm: main theorem}\label{sec: proof of main theorem}}
%---------------------------------------------------------------------------------------------------------------------------------------------------------------

The lower bound in \eqref{eq: main estimate} is a direct consequence of the Birman--Schwinger principle \eqref{eq: Birmann Schwinger principle} together with Corollaries \ref{cor: analysis of Phi} and \ref{cor: final bound for Psi(lambda)}. As the upper bound was already discussed in Section \ref{sec: BS principle and upper bound}, we have completed the proof of Theorem \ref{thm: main theorem}.

%---------------------------------------------------------------------------------------------------------------------------------------------------------------
\section{Proof of Lemma \ref{lem: asymptotic form of G} \label{sec: asymptotic form of G}}
%---------------------------------------------------------------------------------------------------------------------------------------------------------------

As a first step we replace $G(q,\tau)$ by 
\begin{align}
\widetilde G(q,\tau)  = \frac{1}{L^2}\sum_{k} \left( \frac{1}{ m   k^2-E_B} - \frac{\xi_{\mu}(k^2)}{\frac{1}{M}(q-k)^2 + k^2 +\tau}\right),
\end{align}
where
\begin{align}
\xi_\mu(s) = \begin{cases} 
& \quad \quad \quad \quad 0 \hspace{2.85cm}  ( s \le \mu ) \\
& \frac{1}{2}\cos\Big( \frac{ \pi( s - \mu )\log \widetilde \mu  }{\mu } \Big) + \frac{1}{2} \quad \quad ( \mu \le s \le \mu + \mu / \log\widetilde \mu ) \\
& \quad \quad \quad \quad 1 \hspace{2.89cm} ( s \ge \mu + \mu / \log\widetilde \mu  ) .
\end{cases}
\end{align}
Compared to $G( q ,\tau)$ we have replaced the characteristic function $\chi_{(\mu,\infty)}(k^2)$ with a smoother cutoff described by $\xi_\mu(k^2)$. The error for this can be controlled by a crude estimate like
\begin{align}
\vert  G ( q , \tau ) - \widetilde G ( q , \tau ) \vert & \le  \frac{1}{L^2} \sum_{k^2 \ge \mu }	 \frac{\chi_{(\mu,\mu + \mu / \log\widetilde \mu  )}(k^2)}{\frac{1}{M}(q -k)^2 + k^2 +\tau } \le C  \frac{\mu}{(\mu + \tau ) \log \widetilde \mu},
\end{align}
which is easily justified by means of \eqref{eq: general sum for k^2 between a mu and b mu}. Next we write $\widetilde G(q,\tau) = L^{-2}\sum_k g (k)$ with
\begin{align}
 g ( k )  = \frac{1}{ m   k^2 + \vert E_B\vert } - \frac{\xi_{\mu} (k^2)}{\frac{1}{M}(q-k)^2 + k^2 +\tau} ,
\end{align}
and apply Poisson's summation formula (see e.g.\ \cite[Section 3.2]{Grafakos2014}) to find
\begin{align}\label{eq: poisson summation formula}
\widetilde  G(q,\tau) - \frac{1}{4\pi^2}  \int_{\mathbb R^2} g(k) \D^2 k  = \frac{1}{4\pi^2}\bigg( \left(\frac{2\pi}{L}\right)^2 \sum_{k\in \kappa \mathbb Z^2}g(k)  - (2\pi)   \widehat g(0) \bigg) = \frac{1}{2\pi} \sum_{\substack{ z\in L\mathbb Z^2 \\ z \neq 0 } }\widehat g(z).
\end{align}
To compute $\widehat g(0)$ we replace $\xi_\mu(s)$ again with $\chi_{(\mu,\infty)}(k^2)$ and estimate the difference 
\begin{align}
\Bigg\vert \, (2\pi) \widehat g(0) - \int_{\mathbb R^2}  \left( \frac{1}{ m   k^2 + \vert E_B\vert } - \frac{\chi_{(\mu,\infty)} (k^2)}{\frac{1}{M}(q-k)^2 + k^2 +\tau} \right) \D ^2 k \, \Bigg\vert \le  C  \frac{\mu}{(\mu + \tau ) \log \widetilde \mu}.
\end{align}
The integral can be evaluated explicitly,
\begin{align}
& \int_{\mathbb R^2}  \left( \frac{1}{ m k^2 + \vert E_B\vert } - \frac{\chi_{(\mu,\infty)} (k^2)}{\frac{1}{M}(q-k)^2 + k^2 +\tau} \right) \D^2 k\nonumber \\
& \hspace{3.5cm}  =  \frac{\pi}{m } \log\left( \frac{ \frac{1 }{M+1} q^2 + \tau + m  \mu }{ \vert E_B \vert } \right)   +  \frac{\pi }{m } \log\left( 1 - \frac{F(q^2 ,\tau)}{2}\right)
\end{align}
where
\begin{align}
F(s,\tau) =  \frac{ \frac{1}{M} s  + \tau + m \mu   }{ \frac{1}{M+1} s  + \tau + m \mu  } \left( 1 - \sqrt{ 1 - \frac{4 s \frac{\mu}{M^2}} {\big( \frac{1}{M}s + \tau + m  \mu \big)^2  } }\right) .
\end{align}
The last expression is not larger than $1+1/M$ such that for $M>1$, we have
\begin{align}
\bigg\vert  \log\left( 1 - \frac{F(q^2 ,\tau)}{2}\right) \bigg\vert  \le C.
\end{align}
It follows that
\begin{align}
\Bigg\vert  \, (2\pi) \widehat g (0) -   \frac{\pi }{ m } \log\left( \frac{ \frac{1}{M+1} q^2 + \tau + m  \mu }{ \vert E_B \vert } \right)  \Bigg\vert \le   C \bigg( 1 +  \frac{\mu}{(\mu + \tau ) \log \widetilde \mu} \bigg). 
\end{align}

Next we need to estimate the right side in \eqref{eq: poisson summation formula}. To this end, write $g(k) = g_1(k)+g_2(k)$ with
\begin{align}
g_1(k) =   \frac{1}{ m  k^2 + \vert E_B\vert } , \quad \quad g_2(k) = \frac{\xi_{\mu}(k^2)}{\frac{1}{M}(q-k)^2 + k^2 +\tau},
\end{align}
and use rotational symmetry, i.e. $\widehat g_i(z)  = \widehat g_i(\vert z\vert e_u)$, $i=1,2$, where $e_u$ denotes the first unit vector in the $(k_u,k_v)$ plane. We can then use integration by parts to compute the Fourier transform for $ z \neq 0$,
\begin{align}
\widehat g_1(z) & = m^{-1}  \int\limits_{-\infty}^{\infty} \D k_u e^{i k_u \vert z\vert } \int\limits_{-\infty}^\infty \D k_v  \frac{1}{   k^2 + \vert E_B \vert  /m   }\notag \\
& = \frac{ 1 }{m ( i \vert z \vert) ^3}   \int\limits_{-\infty}^{\infty} \D k_u \left( \frac{\partial ^3}{\partial k_u ^3}e^{i k_u \vert z\vert } \right)  \int\limits_{-\infty}^\infty \D k_v  \frac{1}{   k^2 + \vert E_B \vert  /m  }\notag \\
& = \frac{1}{i m \vert z \vert^3}   \int\limits_{-\infty}^{\infty} \D k_u  e^{i k_u \vert z\vert }    \int\limits_{-\infty}^\infty \D k_v   \frac{18 k_u (k_v^2 + \vert E_B\vert /m   ) - 24 k_u^3}{ (  k^2 + \vert E_B \vert /m )^4  }  .
\end{align}
Of the last expression we estimate the absolute value to get
\begin{align}\label{eq: bound for g_1(z)}
\vert \widehat g_1(z)\vert  & \le \frac{1}{m\vert z \vert^3} \int \D ^2k \left( \frac{18 \vert k \vert }{ (  k^2 + \vert E_B \vert /m  )^3  } + \frac{24 \vert k \vert^3 }{ (  k^2 + \vert E_B \vert /m  )^4  } \right) \le \frac{C}{\vert z \vert^3 \vert E_B\vert^{3/2}  }.
\end{align}
The bound for $\vert \widehat g_2(z)\vert $ works similarly but is slightly more cumbersome. We start again with
\begin{align}\label{eq: partial integration for  g_2(z)}
\widehat g_2(z)  & = \frac{1}{im\vert z\vert^3\vert E_B\vert^{3/2} }  \int\limits_{-\infty}^{\infty} \D k_u e^{i k_u \vert z\vert } \int\limits_{-\infty}^{\infty} \D k_v \frac{\partial^3}{\partial k_u^3} \left( \frac{\xi_{\mu}(k^2) \vert E_B\vert^{3/2}}{\frac{1}{M} (q-k)^2 + k^2 +\tau} \right) 
\end{align}
for which we need to compute the different derivatives. A straightforward computation shows
\begin{align}
\Big\vert \frac{\partial^n }{\partial k_u^n} \xi_\mu(k^2) \Big\vert & \le C \frac{  (\log\widetilde \mu)^n}{\mu^{n/2} } \chi_{(\mu,\mu+\mu/\log \widetilde \mu)}(k^2),\quad n\in \{1,2,3\}.
\end{align}
%\begin{align}
%\Big\vert \frac{\partial}{\partial k_u } \xi_\mu(k^2) \Big\vert & \le C \frac{  \log\widetilde \mu }{\mu^{1/2} } \chi_{(\mu,\mu+\mu/\log \widetilde \mu)}(k^2),\\
%\Big\vert \frac{\partial^2}{\partial k_u^2} \xi_\mu(k^2) \Big\vert & \le C  \frac{ ( \log\widetilde \mu)^2 }{\mu^{n/2} }  \chi_{(\mu,\mu+\mu/\log \widetilde \mu)}(k^2),\\
%\Big\vert \frac{\partial^3}{\partial k_u^3} \xi_\mu(k^2) \Big\vert & \le C  \frac{ ( \log\widetilde \mu)^3 }{\mu^{3/2} }  \chi_{(\mu,\mu+\mu/\log %\widetilde \mu)}(k^2).
%\end{align}
Abbreviating the denominator as $D(k)  =  \frac{1}{M} (q-k)^2 + k^2 +\tau  $ it is not difficult to verify
\begin{align}
\chi_{(\mu,\mu+\mu/\log \widetilde \mu)}(k^2) \bigg\vert \frac{\partial }{\partial k_u} \frac{1}{D(k)}   \bigg\vert & \le C \bigg(  \frac{1}{\mu^{1/2} (k^2+ \tau )}  +  \frac{\mu^{1/2}}{(k^2+ \tau )^{2} }\bigg)  \label{eq: bound for D(k)} ,\\
\chi_{(\mu,\mu+\mu/\log \widetilde \mu)}(k^2) \bigg\vert \frac{\partial^2 }{\partial k_u^2} \frac{1}{D(k)}  \bigg\vert & \le C \bigg(  \frac{ 1}{ (k^2+ \tau )^{2}} + \frac{\mu}{ (k^2+ \tau )^{3}} \bigg ) ,\\
\chi_{(\mu,\mu+\mu/\log \widetilde \mu)}(k^2) \bigg\vert \frac{\partial^3}{\partial k_u^3} \frac{1}{D(k)}  \bigg\vert & \le C \bigg(  \frac{1}{\mu^{1/2}  (k^2+ \tau )^{2} } + \frac{\mu^{1/2}}{ (k^2+ \tau )^{3}} + \frac{\mu^{3/2}}{ (k^2+ \tau )^{4} } \bigg) .
\end{align}
To illustrate this for the first line, we compute
\begin{align}
\bigg\vert \frac{\partial }{\partial k_u} \frac{1}{D(k)}   \bigg\vert =  \bigg\vert  \frac{1}{D(k)^2} \bigg( \frac{2}{M}(k_u-q_u) + 2 k_u \bigg) \bigg\vert  \le C \bigg( \frac{\vert k - q \vert }{D(k)^2} + \frac{\vert k  \vert }{D(k)^2}\bigg) 
\end{align}
and use $D(k)\ge \frac{1}{M} (k-q)^2$ and $k^2 \le 2  \mu$ in combination with a balanced Cauchy-Schwarz estimate. This leads to
\begin{align}
\bigg\vert \frac{\partial }{\partial k_u} \frac{1}{D(k)}   \bigg\vert \le  C \bigg( \frac{1}{D(k)^{3/2}} + \frac{\sqrt \mu}{   D(k)^2}\bigg)  \le C  \bigg( \frac{1}{\sqrt \mu D(k)} + \frac{\sqrt \mu}{   D(k)^2}\bigg) 
\end{align}
from which the bound in \eqref{eq: bound for D(k)} follows by $D(k) \ge k^2+\tau$. The other two lines are obtained in close analogy.

Summing up the different combinations we obtain
\begin{align}
& \bigg\vert \frac{\partial^3}{\partial k_u^3} \bigg( \frac{\xi_{\mu}(k^2) \vert E_B\vert ^{3/2}}{\frac{1}{M} (p-k)^2 + k^2 +\tau} \bigg) \bigg\vert \notag \\
& \quad \quad \le \frac{C(\log \widetilde \mu)^{3}}{ \widetilde \mu^{3/2}  } \bigg(  \chi_{(\mu , \mu + \mu/\log \widetilde \mu)}(k^2)  \frac{1}{k^2 +\tau } + \chi_{(\mu,\infty)}(k^2) \sum_{j=2}^4 \frac{\mu^{j-1}}{(k^2+\tau)^{j}} \bigg) 
\end{align}
by which we can estimate the integral
\begin{align}
\int_{\mathbb R^2} \bigg\vert \frac{\partial^3}{\partial k_u^3} \bigg( \frac{\xi_{\mu}(k^2)\vert E_B\vert ^{3/2} }{\frac{1}{M} (q-k)^2 + k^2 +\tau} \bigg) \bigg \vert\D ^2 k 
% \le \frac{C(\log \widetilde \mu)^{3}}{  \widetilde  \mu^{3/2}} \bigg( 1 + \frac{\mu}{\mu + \tau} \bigg) ^{3} \notag \\
&  \le C  \bigg( 1+\frac{\mu}{(\mu + \tau) \log \widetilde \mu} \bigg) ^{3}.
\end{align}
Together with \eqref{eq: bound for g_1(z)} and \eqref{eq: partial integration for  g_2(z)}, this gives
\begin{align}
\vert \widehat g(z) \vert  \le \frac{C}{\vert z\vert^3 \vert E_B \vert^{3/2}  }  \bigg( 1+\frac{\mu}{(\mu + \tau) \log \widetilde \mu} \bigg) ^{3}.
\end{align}
The remaining series can be bounded as
\begin{align}
\sum_{\substack{ z\in L\mathbb Z^2 \\ z \neq 0 } } \vert z\vert^{-3} \vert E_B \vert^{-3/2}  = ( L \vert E_B\vert^{1/2})^{-3} \sum_{\substack{ z \in \mathbb Z^2 \\ z\neq 0 } } \vert z \vert^{-3} \le C
\end{align}
because of $ L^2  \vert E_B\vert \ge 1$ and
\begin{align}
\sum_{\substack{ z \in \mathbb Z^2 \\ z\neq 0 } }\frac{1}{\vert z \vert^3} & = \sum_{n,m \ge 1} \frac{4}{(n^2+m^2)^{3/2}} + \sum_{n\ge 1 } \frac{4}{n^3}  \le \int\limits_1^\infty \bigg( \frac{8\pi }{s^2}  + \frac{4}{s^3} \bigg) \D s \le C
\end{align}
by the integral test of convergence. 

We conclude that the absolute value of the right side in \eqref{eq: poisson summation formula} is bounded from above by
\begin{align}
 \frac{1}{2\pi} \sum_{\substack{ z\in L\mathbb Z^2 \\ z \neq 0 } } \vert \widehat g(z) \vert \le C   \bigg( 1+\frac{\mu}{(\mu + \tau) \log \widetilde \mu} \bigg) ^{3}.
\end{align}
Hence the proof of the lemma is complete.

%---------------------------------------------------------------------------------------------------------------------------------------------------------------
\appendix
%---------------------------------------------------------------------------------------------------------------------------------------------------------------

\section{Replacing sums by integrals}

For a short proof of the following lemma, see \cite[Appendix B]{LM_InfiniteMass}.

\begin{lemma}\label{lemma: replacing sums by integrals} (a) Let $f:[0,\infty)\to [0,\infty)$ be monotonically decreasing. Then,
\begin{align}
\Bigg| \frac{1}{L^2} \sum_{k} f(k^2) - \frac{1}{2\pi}\int\limits_0^\infty f(t^2) t\ \D t \Bigg| \le \frac{2}{\pi L} \int\limits_0^\infty f(t^2)  \D t + \frac{3f(0)}{L^2} .
\end{align}
(b) Let $m\ge 0$ and $f:[m,\infty)\to [0,\infty)$ be monotonically decreasing. Then,
\begin{align}
\Bigg| \frac{1}{L^2} \sum_{k^2\ge m} f(k^2) - \frac{1}{2\pi}\int\limits_{\sqrt{m}}^\infty f(t^2) t\ \D t \Bigg| \le \frac{2}{\pi L} \int\limits_{\sqrt{m}}^\infty f(t^2)  \D t + \Big( \frac{4 \sqrt m}{\pi L} + \frac{6}{L^2} \Big) f(m). 
\end{align}
\end{lemma}

\section{Completing the proof of Lemma \ref{lem: perturbed polaron equation}\label{App: completing the proof of perturbed polaron equation}}

It remains to analyze the condition $\mathcal{F}(T-\lambda,r)\ge 0$. For that we approximate $\mathcal{F}(T-\lambda,r)$ by $\mathcal{F}^{(n)}(T-\lambda,r)$ where the operator $\mathcal{F}^{(n)}(T-\lambda,r)$ arises by replacing the 
function $G(0,\tau)$ in \eqref{eq: definition of mathcal F} by 
\begin{align}
  G^{(n)}(0,\tau) = \frac{1}{L^2} \sum_{k^2 \leq n} \left( \frac{1}{k^2 - E_B} - 
\frac{\chi_{(\mu,\infty)}(k^2)}{k^2 + \tau} \right).
\end{align}
Note that $G^{(n)}(0,\tau) \to G(0,\tau)$ as $n \to \infty$ for every $\tau > -\mu$. Thus, $G^{(n)}(0,T-\lambda) \psi \to G_\mu(0,T-\lambda) \psi$ as $n 
\to \infty$ for every $\psi \in D$ (recall that $D\subset \HH_{N(\mu)-1}$ is the set of all anti-symmetric product states, see \eqref{eq: definition of total set D}). Since $D$ forms a total set of eigenstates of 
$G(0,T-\lambda)$ on $\HH_{\Nmu-1}$, its linear hull $\text{lin}D \subseteq \HH_{\Nmu-1}$ is a domain of essential self-adjointness for this 
operator. Furthermore, $G^{(n)}(0,T-\lambda) \geq G^{(n)}(0,-\mu - e_{\rm P}(\mu,E_B))$ and thus, by the convergence of $G^{(n)}(0,\tau)$ and the fact that $G(0,-\mu-e_{\rm P}(\mu))\ge C \log \widetilde \mu$, it follows that there is a $c > 0$ such that $G^{(n)}(0,T-\lambda) - r  > c $ for $n$ large enough. Hence as $n \to \infty$, 
$(G^{(n)}(0,T-\lambda) -r )^{-1} \to (G(0,T-\lambda) - r )^{-1}$ and $\mathcal{F}^{(n)}(T-\lambda,r) \to 
\mathcal{F}(T-\lambda,r)$ strongly.

Using the pull-through formula \eqref{eq: pull through formula}, we can write
\begin{align}
\mathcal{F}^{(n)}(T-\lambda,r) & =    T - \lambda - \frac{1}{L^2} \sum_{k^2 \leq \mu} (G^{(n)}(0,T - k^2 - \lambda) - r )^{-1}\notag\\
& \hspace{1cm}   + \frac{1}{L^2} \sum_{k^2,l^2 \leq \mu} \!\! a_k (G^{(n)}(0,T - k^2 - l^2 - \lambda) - r )^{-1} a_l^* \label{T_lambda_n}
\end{align}
on $\HH_{\Nmu}$. Assuming that the last term in \eqref{T_lambda_n}, which we call $\mathcal{P}^{(n)}(T-\lambda,r)$ in the following, is a positive operator on 
$\HH_{\Nmu}$, we obtain 
\begin{align}
   \mathcal{F}^{(n)}(T-\lambda,r) \geq E_0(\mu) - \lambda - \frac{1}{L^2} \sum_{k^2 \leq \mu} (G^{(n)}(0,E_0(\mu) - k^2 - \lambda) - r )^{-1},
\end{align}
since $T \geq E_0(\mu)$ on $\HH_{\Nmu}$. In view of \eqref{eq: lower bound for pol operator}, this completes the proof of Lemma \ref{lem: perturbed polaron equation}.

It remains to show $\mathcal{P}^{(n)}(T-\lambda,r) \geq 0$. For $\psi \in \HH_{\Nmu}$,
\begin{align}
   L^2 \, \sprod{\psi}{\mathcal{P}^{(n)}(T-\lambda, r ) \psi} 
   &= \int\limits_0^\infty   \sum_{k^2,l^2 \leq \mu} \sprod{\psi}{a_k \: \exp(-t \, [ G^{(n)}(0,T - k^2 - l^2 - 
\lambda) - r ) \: a_l^* \psi}\, \D t \nonumber \\
   &= \int\limits_0^\infty  \: \exp(-t [L^{-2} \sum_{p^2 \leq n} \tfrac{1}{p^2 - E_B} \!-\!  r ]) \, \mathcal{I}^{(n)}(t)\, \D t
\end{align}
with
\begin{align}
   \mathcal{I}^{(n)}(t) = \sum_{k^2,l^2 \leq \mu} \sprod{\psi}{a_k 
\prod\limits_{\mu < q^2 \leq n} \!\! \exp(t (q^2 + T - k^2 - l^2 - \lambda)^{-1}) \:\, a_l^* \psi}.
\end{align}
We show that $\mathcal{I}^{(n)}(t) \geq 0$ for all $t \in [0,\infty)$ and $n \in \N$. Note that the product in the definition of $\mathcal{I}^{(n)}(t)$ has 
only finitely many factors, because $A_n = \{ q \in \frac{2\pi}{L} \Z^2 \: | \: \mu < q^2 \leq n \}$ is a finite set. 
We consider the exponential series and obtain
\begin{align}
   \mathcal{I}^{(n)}(t) = \sum_{k^2,l^2 \leq \mu} \sprod{\psi}{a_k \prod\limits_{q \in A_n} \!\! \left( \sum\limits_{m=0}^\infty \frac{t^m}{m!} \, 
\frac{1}{(q^2 + T - k^2 - l^2 - \lambda)^m} \right) a_l^* \psi}.
\end{align}
By the absolute convergence of the exponential series, we can rearrange the product of series to get
\begin{align}
   \mathcal{I}^{(n)}(t) = \sum_{m: A_n \to \N_0} \:\: \sum_{k^2,l^2 \leq \mu} \sprod{\psi}{a_k \prod\limits_{q \in A_n} \!\! \left( \frac{t^{m(q)}}{m(q)!} 
\,
\frac{1}{(q^2 + T - k^2 - l^2 - \lambda)^{m(q)}} \right) a_l^* \psi},
\end{align}
where we sum over all $\N_0$-valued functions $m$ on the finite set $A_n$. Note that the factor in parentheses indexed 
by $q$ is equal to $1$ if $m(q) = 0$. For all factors with $m(q) \neq 0$, we use the identity
\begin{align}
   \frac{1}{a^\tau} = \frac{1}{c_\tau}   \int\limits_0^\infty \! \D s \: e^{-a s^{1/\tau}} \qquad \text{with} \qquad 
c_\tau = 
\int\limits_0^\infty \! \D s \: e^{- s^{1/\tau}}
\end{align}
for $a,\tau > 0$ to rewrite each of the summands in the $m$-sum as
\begin{align}
 &\sum_{k^2,l^2 \leq \mu} \:\: \prod\limits_{\substack{q \in A_n \\ m(q) \neq 0}} \: \left( \frac{ t^{m(q)}}{m(q)!\, c_{m(q)}} 
\int\limits_0^\infty \! \D s_q \right) \sprod{\psi}{a_k \prod\limits_{\substack{p \in A_n \\ m(p) \neq 0}} \!\! e^{-(p^2 + T - k^2 - l^2 - \lambda)s_p^{1/m(p)}} 
a_l^* \psi} \notag \\
 &=\prod\limits_{{\substack{q \in A_n \\ m(q) \neq 0}}} \left( \frac{ t^{m(q)}}{m(q)!\, c_{m(q)} } \int\limits_0^\infty \! \D s_q \right) \norm{\sum\limits_{k^2 \leq 
\mu} \prod\limits_{{\substack{p \in A_n \\ m(p) \neq 0}}} e^{-\frac{1}{2}(p^2 + T - 2k^2 - \lambda)s_p^{1/m(p)}} a_k^* \psi }^2.
\end{align}
This yields $\mathcal{I}^{(n)}(t) \geq 0$ and thus $\sprod{\psi}{\mathcal{P}^{(n)}(T-\lambda,r) \psi} \geq 0$ for all $n \in \N$.

%---------------------------------------------------------------------------------------------------------------------------------------------------------------
\bigskip\vspace{3mm}\noindent
\textbf{Acknowledgements.}\medskip

\noindent I am very grateful to Ulrich Linden for introducing me to the Fermi polaron as well as for his contributions to this project in its early stage.

%---------------------------------------------------------------------------------------------------------------------------------------------------------------

\end{spacing}

%---------------------------------------------------------------------------------------------------------------------------------------------------------------
%---------------------------------------------------------------------------------------------------------------------------------------------------------------
\vspace{7.5mm}
\noindent\textit{E-mail:} \texttt{mitrouskas@mathematik.uni-stuttgart.de} 

%---------------------------------------------------------------------------------------------------------------------------------------------------------------
%---------------------------------------------------------------------------------------------------------------------------------------------------------------

\end{document}